\documentclass[11pt]{article}
\usepackage{fullpage}
\usepackage{times}
\usepackage{amsmath,amsfonts,amssymb,amsthm}
\usepackage{enumerate}
\usepackage{xcolor}
\usepackage{tikz}
\usetikzlibrary{decorations.pathreplacing}

\newtheorem{definition}{Definition}
\newtheorem{theorem}{Theorem}
\newtheorem{lemma}{Lemma}
\newtheorem{claim}{Claim}
\newtheorem{fact}{Fact}

\newtheoremstyle{restate}{}{}{\itshape}{}{\bfseries}{~(restated).}{.5em}{\thmnote{#3}}
\theoremstyle{restate}
\newtheorem*{restate}{}

\newcommand{\E}{\mathop{\mathbb{E}}}
\newcommand{\SUM}{\mathrm{SUM}}
\newcommand{\rank}[1]{\texttt{rank}($#1$)}
\def\sgn{\mathop{\rm sgn}}
\allowdisplaybreaks

\newcommand{\cR}{\mathcal{R}}
\newcommand{\cK}{\mathcal{K}}
\newcommand{\bs}{\mathbf{s}}
\newcommand{\bj}{\mathbf{j}}
\newcommand{\bOne}{\mathbf{1}}
\newcommand{\bT}{\mathbf{T}}

\newcommand{\bE}{\mathbf{E}}

\title{Optimal Succinct Rank Data Structure via \\ Approximate Nonnegative Tensor Decomposition}
\author{Huacheng Yu\thanks{Harvard University. \texttt{yuhch123@gmail.com}. Supported in part by ONR grant N00014-15-1-2388, a Simons Investigator Award and NSF Award CCF 1715187.}}
\date{}

\begin{document}

\setcounter{page}{0}
\maketitle
\thispagestyle{empty}
\begin{abstract}
Given an $n$-bit array $A$, the succinct rank data structure problem asks to construct a data structure using space $n+r$ bits for $r\ll n$, supporting rank queries of form $\rank{$u$}=\sum_{i=0}^{u-1} A[i]$.
In this paper, we design a new succinct rank data structure with $r=n/(\log n)^{\Omega(t)}+n^{1-c}$  and query time $O(t)$ for some constant $c>0$, improving the previous best-known by P{\v a}tra{\c s}cu~\cite{Pat08}, which has $r=n/(\frac{\log n}{t})^{\Omega(t)}+\tilde{O}(n^{3/4})$ bits of redundancy.
For $r>n^{1-c}$, our space-time tradeoff matches the cell-probe lower bound by P{\v a}tra{\c s}cu and Viola~\cite{PV10}, which asserts that $r$ must be at least $n/(\log n)^{O(t)}$.
Moreover, one can avoid an $n^{1-c}$-bit lookup table when the data structure is implemented in the cell-probe model, achieving $r=\lceil n/(\log n)^{\Omega(t)}\rceil$.
It matches the lower bound for the full range of parameters.

En route to our new data structure design, we establish an interesting connection between succinct data structures and approximate nonnegative tensor decomposition.
Our connection shows that for specific problems, to construct a space-efficient data structure, it suffices to approximate a particular tensor by a sum of (few) nonnegative rank-$1$ tensors.
For the rank problem, we explicitly construct such an approximation, which yields an explicit construction of the data structure.
\end{abstract}
\newpage

\section{Introduction}

Given an array $A[0..n-1]$ of bits, the partial sums problem (a.k.a, the rank problem) asks to preprocess $A$ into a data structure using as little space as possible, supporting queries of form $\rank{$u$}=\sum_{i=0}^{u-1} A[i]$ efficiently.
One trivial solution is to explicitly write down all prefix sums, which uses $n$ \emph{words} of space and constant query time.
In succinct data structures, one seeks data structures using space close to the information theoretical limit, $n$ \emph{bits} for partial sums, with an efficient query time.

Succinct rank data structures are central building blocks in many succinct data structure problems with a rich history~\cite{Jacobson89,CM96,Munro96,Clark97,MRR98,RRR02,GMR06,Golynski07,Pat08}.
Jacobson~\cite{Jacobson89}, Clark and Munro~\cite{CM96} gave the first succinct rank data structures using $n+o(n)$ space with constant query time.
After a series of improvements~\cite{Munro96,MRR98,RRR02}, Golynski et al.~\cite{GGGRR07} achieved space $n+O(\frac{n\log\log n}{\log^2 n})$ for constant query time.
Later, the seminal paper ``Succincter'' by P\v{a}tra\c{s}cu~\cite{Pat08} proposed a data structure using space $n+n/(\frac{\log n}{t})^t+\tilde{O}(n^{3/4})$ and query time $O(t)$, showing that the redundant bits can be any $n/\mathrm{poly}\log n$ when query time is constant.

Lower bounds for this problem also have received attention from researchers in the area~\cite{Miltersen05,GM07,Golynski07b,PV10}.
Most lower bounds are for ``systematic encodings.''
In systematic encoding, we are given an input that is stored explicitly in the raw form, and we may then build a (sublinear) auxiliary data structure, which will be stored on the side.
The query algorithm has access to both the raw input and the auxiliary data structure.
Golyski~\cite{Golynski07b} showed a space lower bound of $n+(n\log t)/t$ for query time $t$, for any systematic encoding of the rank problem.
For general data structures, P{\v a}tra{\c s}cu and Viola~\cite{PV10} proved a space lower bound of $n+n/w^{O(t)}$ for query time $t$ in the cell-probe model with word-size $w$ (implying the same RAM lower bound).
In the standard regime where $w=\Theta(\log n)$, this space lower bound matches the ``Succincter'' upper bound for constant query times.

However, if one insists on, say $c_0\log n/\log\log n$ query time, for sufficiently large constant $c_0$, then the best-known data structure occupies at least $n+n^{1-o(1)}$ bits of space on a worst-case $n$-bit input, whereas the state-of-the-art lower bound does not even rule out an $(n+1)$-bit data structure!
Closing this gap is referred to as ``a difficult problem'' in~\cite{PV10}.

Interestingly, we show such an $(n+1)$-bit data structure \emph{does exist}, if we allow arbitrary $O(w)$-bit word operations.
\begin{theorem}[informal]\label{thm_upper_cell}
	Given an $n$-bit array, one can construct a data structure using
	\[
		n+\lceil\frac{n}{w^{\Omega(t)}}\rceil
	\]
	bits of memory supporting rank queries in $O(t)$ time, where $w\geq \Omega(\log n)$ is the word-size, assuming the data structure can perform arbitrary $O(w)$-bit word operations.
\end{theorem}

By applying a standard trick for self-reducible problems and storing lookup tables for the necessary $O(w)$-bit word operations, this data structure can also be implemented in word RAM with an $n^{1-c}$ space overhead.
\begin{theorem}[informal]\label{thm_upper_ram}
	Given an $n$-bit array, one can construct a data structure using
	\[
		n+ \frac{n}{(\log n)^{\Omega(t)}}+n^{1-c}
	\]
	bits of memory supporting rank queries in $O(t)$ time, in a word RAM with word-size $\Theta(\log n)$, for some universal constant $c>0$.
\end{theorem}

In particular, for query time $t=c_0\log n/\log\log n$ for sufficiently large $c_0$, Theorem~\ref{thm_upper_cell} gives us a data structure using only $n+1$ bits of memory.
Moreover, our new cell-probe data structure matches the P\v{a}tra\c{s}cu-Viola lower bound for any bits of redundancy, up to a constant factor in query time.
It settles the space-time tradeoff for succinct rank in the cell-probe model.
One may also observe if only exactly $n$ bits of memory are allowed, then there is nothing one can do beyond storing $A$ explicitly.
It is because in this case, the data structure has to be a bijection between $A$ and the memory contents.
In particular, even to verify whether $\rank{$n$}=0$, one has to check if the memory content corresponds to the exact all-zero input, which requires a linear scan.
However, when $n+1$ bits are allowed, half of the $2^{n+1}$ memory configurations may be unused, which could potentially facilitate the query algorithm.

\bigskip

En route to our new succinct rank data structure construction, we establish an interesting connection between succinct data structures and \emph{approximate nonnegative tensor decomposition}.
Tensors are generalizations of matrices.
An order-$B$ tensor can be viewed as a $B$-dimensional array.
Analogous to the rank of a matrix, a tensor $\bT$ has rank $1$ if its entries can be written as
\[
	\bT_{x_1,\ldots,x_B}=a^{(1)}_{x_1}\cdot a^{(2)}_{x_2}\cdots a^{(B)}_{x_B},
\]
for vectors $a^{(1)},\ldots,a^{(B)}$.
The rank of a tensor $\bT$ is the minimum number of rank-$1$ tensors, of which $\bT$ can be expressed as a sum.
The nonnegative rank of a nonnegative tensor further requires each rank-$1$ tensor (or equivalently, each $a^{(i)}$) to be nonnegative, hence is at least as large as the rank.
Given a nonnegative tensor $\bT$ and parameters $r,\epsilon$, the problem of \emph{approximate nonnegative tensor decomposition} asks to find $r$ nonnegative rank-$1$ tensors $\bT_1,\ldots,\bT_r$ such that $\|\bT-(\bT_1+\cdots+\bT_r)\|\leq \epsilon$ under certain norm, if exists.

\paragraph{Connections to tensor decomposition.} 
As we mentioned in the beginning, explicitly storing all prefix sums takes too much space.
One inherent reason is that the prefix sums are very correlated.
Denote by $T_i$ the number of ones in the first $i$ bits.
One may verify that for uniformly random inputs (which maximizes the input entropy) and $i<j$, $I(T_i; T_j)\approx \frac{1}{2}\log \frac{j}{j-i}$.
Even just storing $T_{100}$ and $T_{110}$ in separate memory words would already introduce a redundancy of $\frac{1}{2}\log 11>1.5$ bits, because the ``same $1.5$ bits of information'' is stored in two different locations.
Hence, in order to achieve low redundancy, only mutually (almost) independent variables could be stored separately.

The key observation used in our new data structure is the following.
Suppose we were to store $B$ (correlated) numbers $y_1,\ldots,y_B\in[n]$, if we could find another variable $\eta$ such that conditioned on $\eta$, $y_1,\ldots,y_B$ become (almost) mutually independent, then one could hope to first store $\eta$, then store these $B$ numbers \emph{conditioned on} $\eta$.
To retrieve one $y_i$, one always first reads $\eta$, which reveals the representation of $y_i$, then reads $y_i$.
This strategy is only possible if the support size of $\eta$ is not too large, and can be encoded using few bits, since its value needs to be retrieved prior to reading any $y_i$.
If we are aiming at constant retrieval time, then the support size of $\eta$ must be at most $2^{O(w)}$.

The joint distribution of $(y_1,\ldots,y_B)\in[n]^B$ can be described by an order-$B$ tensor $\bT$ of size $n^B$, where the entries describe the probability masses.
Any nonnegative rank-$1$ tensor of this size would correspond to an independent distribution.
Suppose we could find $r$ nonnegative rank-$1$ tensors $\bT_1,\ldots,\bT_r$ such that
\[
	\|\bT-(\bT_1+\cdots+\bT_r)\|_1\leq \epsilon.
\]
This would imply that $\bT$ can approximately be viewed as a convex combination of $r$ independent distributions.
Let $\eta$ indicate which independent distribution we are sampling from, then $y_1,\ldots,y_B$ become conditionally independent conditioned on $\eta$ (except for a small probability of $\epsilon$).
More importantly, the support size of $\eta$ is equal to $r$.
If such decomposition is possible for $r=2^{O(w)}$ and $\epsilon=1/\mathrm{poly}\, n$, then we will have hope to store $(y_1,\ldots,y_B)$ with constant retrieval time and (negligible) redundancy of $1/\mathrm{poly}\, n$.

\paragraph{Computing tensor decomposition.}
In general, nonnegative tensor decomposition is computationally difficult. 
Even for tensor order $B=2$ (i.e., nonnegative \emph{matrix} factorization), any algorithm with running time subexponential in $r$ would yield a subexponential time algorithm for 3-SAT~\cite{AGKM16,Moitra16}, breaking Exponential Time Hypothesis.
Fortunately, the tensors that we obtain from the data structure problem are not arbitrary.
For the rank problem, the values $\bT_{x_1,\ldots,x_B}$ are relatively smooth as a function of $(x_1,\ldots,x_B)$.
Given such a tensor, we may partition it into small cubes, and approximate the values within each cube by a low degree polynomial $P(x_1,\ldots,x_B)$.
The key observation here is that the tensor corresponding to a monomial $x_1^{e_1}\cdots x_B^{e_B}$ has rank $1$.
If $P$ has low degree, thus has a small number of monomials, then the tensor restricted to the cube must have low approximate rank.
To make this approximation nonnegative, we apply the following transformation to a negative monomial $-x^ay^b$ (plus a large constant):
\begin{equation}\label{eqn_neg}
	M^{a+b}-x^ay^b\equiv \frac{1}{2}(M^a-x^a)(M^b+y^b)+\frac{1}{2}(M^a+x^a)(M^b-y^b).
\end{equation}
When $x, y\in[0, M]$, both terms on the RHS become nonnegative.
One may also generalize this equation to monomials with more than two variables.
The polynomial obtained from each small cube has a sufficiently large constant term so that all negative monomials can be transformed simultaneously via the above equation.
Finally, by summing up the approximations within each cube, we obtain a low rank nonnegative approximation for the whole tensor.

\subsection{Organization of the paper}
In Section~\ref{sec_prelim}, we give preliminary and define notations.
In Section~\ref{sec_overview}, we give an overview of the new succinct rank data structure, as well as a summary of \cite{Pat08}.
In Section~\ref{sec_upper}, we present our data structure construction.
Finally, we conclude with discussions and open questions in Section~\ref{sec_concl}.

\section{Preliminary and Notations}\label{sec_prelim}
\subsection{Notations}
	Let $a, b\in\mathbb{R}$ and $b\geq 0$, denote by $a\pm b$ the set $[a-b, a+b]$.
	Similarly, denote by $c(a\pm b)$ the set $[c(a-b), c(a+b)]$.
	Throughout the paper, $\log x=\log_2 x$ is the binary logarithm.
\subsection{Spillover representation}
One important technique in succinct data structures is the \emph{spillover representation}, formally introduced by P\v{a}tra\c{s}cu~\cite{Pat08}.
Similar ideas also appeared in Munro et al.~\cite{MRRR03} and Golyski et al.~\cite{GGGRR07}.
It allows one to use ``fractional bits'' of memory in each component of the data structure construction so that the fractions can be added up before rounding.

A data structure is said to use $m$ bits of memory with spillover size $K$, if it can be represented by a pair consisting of $m$ bits and a number (called the \emph{spillover}) in $[K]$ (usually $K\leq 2^{O(w)}$).
We may also say that the data structure uses space $[K]\times \{0,1\}^m$.
At the query time, we assume that the spillover $k\in[K]$ is given for free, and \emph{any} consecutive $w$ bits of the representation can be accessed in constant time.
These $m$ bits are usually stored explicitly in the memory.
Hence, any consecutive $w$ bits can be retrieved by reading one or two consecutive words.

Intuitively, such representations use ``$m+\log K$ bits'' of space.
It avoids the problem of ``rounding to the next integer bit'' in designing a subroutine, as now we can round up the spillover, which wastes much less space: $\log (K+1)$ and $\log K$ only differ by $\log (1+1/K)\approx 1/K$ bit.
Setting $K=\Omega(n^2)$ makes it negligible.

\subsection{The cell-probe model}
The cell-probe model proposed by Yao~\cite{Yao78} is a powerful non-uniform computational model, primarily used in data structure lower bound proofs.
In the cell-probe model, we only measure the number of memory accesses.
The memory is divided into \emph{cells} of $w$ bits.
The data structure may read or write the content of a memory cell by \emph{probing} this cell.
In each step of the query algorithm, it may probe one memory cell (or in the last step, it returns the answer to the query), based on all information it has obtained so far, including the query and all previous contents the algorithm has seen.
The running time is defined to be the number of memory cells probed.
We assume computation is free.

\section{A Streamlined Overview}\label{sec_overview}
In this section, we overview our new succinct rank data structure.
To avoid heavy technical details, we are going to present the data structure in a stronger model, where each word may store ``$O(w)$ bits of information about the data'', rather than an exact $w$-bit string.
We refer to this model of computation as the \emph{information cell-probe model}.
The data structure for standard word RAM can be found in Section~\ref{sec_upper}.

Let us first assume that the input $n$-bit array $A$ is uniformly random.\footnote{Uniform distribution maximizes the input entropy, making the same space bound more difficult to achieve (at least in expectation). The data structure in Section~\ref{sec_upper} for standard word RAM does not rely on this assumption, and the space bound there is for worst-case input.}
The information that each memory word (content) $C_i$ reveals about the input $I(C_i;A)$ is measured under this input distribution.
This mutual information, in some sense, measures ``the space usage of word $C_i$.''
We will focus on the case where $w=\Theta(\log n)$, and the goal is to design a data structure that
\begin{itemize}
	\item stores no more than $w$ bits of information in every memory cell $C_i$: $I(C_i;A)\leq w$,
	\item uses no more than ``$n+1$ bits'' of space: $\sum_i I(C_i;A)\leq n+1$, and
	\item supports rank queries \rank{u} in $O(\log_w n)=O(\log n/\log\log n)$ time.
\end{itemize}

The top-level of the data structure is similar to~\cite{Pat08}, which is a standard range tree with branching factor $B$, here for $B=w^{1/3}$.
For simplicity, we assume $n$ is a power of $B$ (and $n=B^t$).
Given an input array $A[0,\ldots,n-1]$, we construct a tree with branching factor $B$ and depth $t$ by recursion.
Each node $v$ at level $i$ in the tree is associated with a subarray $A_v$ of length $B^{t-i}$.
For $j\in [B]$, the $j$-th child of $v$ is associated with the subarray containing $((j-1)B^{t-i-1}+1)$-th to $(jB^{t-i-1})$-th element of $A_v$.
Figure~\ref{fig_datastr} presents the top-level structure (a standard range tree).
\begin{figure}[!ht]
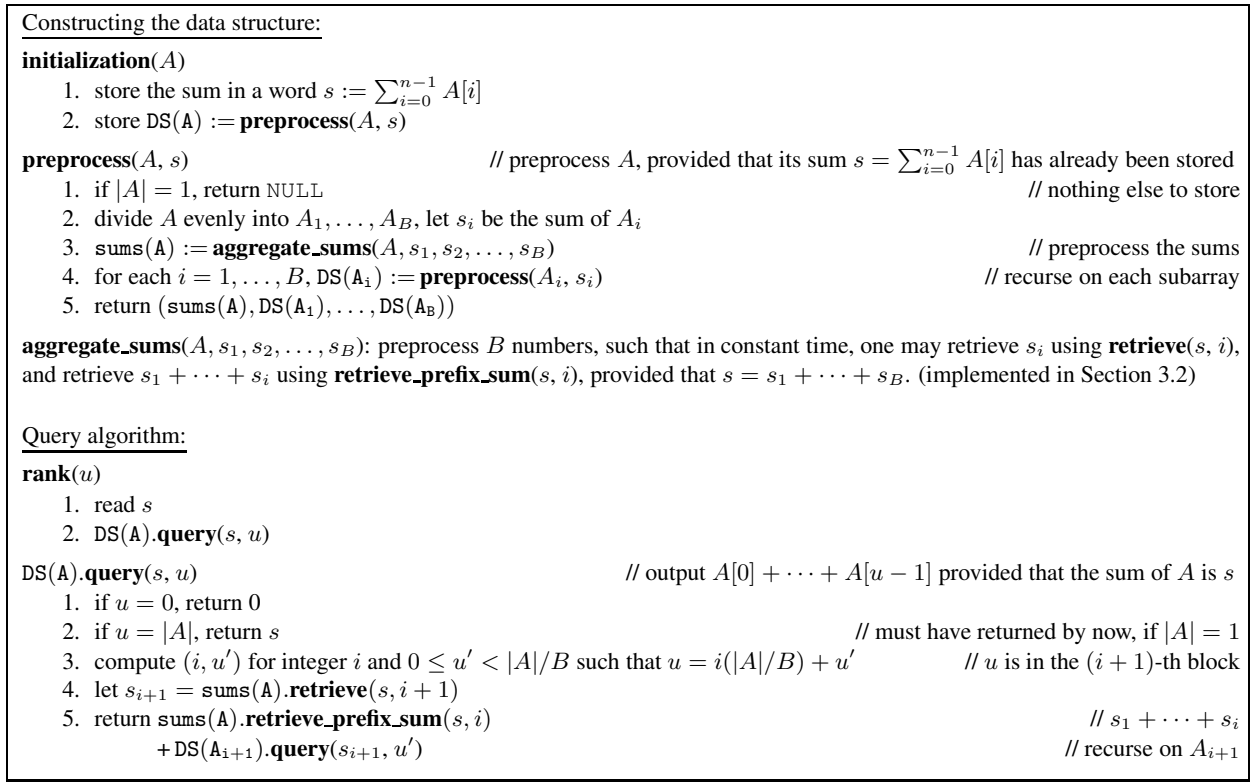

\begin{center}
\fbox{                                                                                                         
{\footnotesize                                                                                                 
\parbox{6.375in} {                                                                                             
\underline{Constructing the data structure:}\\[0.05in]
\textbf{initialization}($A$)
\vspace{-.05in}\begin{enumerate}
\addtolength{\itemsep}{-4pt}
\item store the sum in a word $s:=\sum_{i=0}^{n-1} A[i]$
\item store $\mathtt{DS(A)}:=$\,\textbf{preprocess}($A$, $s$)
\end{enumerate} 

\textbf{preprocess}($A$, $s$) \hfill // preprocess $A$, provided that its sum $s=\sum_{i=0}^{n-1} A[i]$ has already been stored
\vspace{-.05in}\begin{enumerate}
\addtolength{\itemsep}{-4pt}
\item if $|A|=1$, return \texttt{NULL} \hfill // nothing else to store
\item divide $A$ evenly into $A_1,\ldots,A_B$, let $s_i$ be the sum of $A_i$
\item $\mathtt{sums(A)}:=$\,\textbf{aggregate\_sums}($A, s_1,s_2,\ldots,s_B$) \hfill // preprocess the sums
\item for each $i=1,\ldots,B$, $\mathtt{DS(A_i)}:=$\,\textbf{preprocess}($A_i$, $s_i$) \hfill // recurse on each subarray
\item return $(\mathtt{sums(A)}, \mathtt{DS(A_1)},\ldots,\mathtt{DS(A_B)})$
\end{enumerate}

\textbf{aggregate\_sums}($A, s_1,s_2,\ldots,s_B$): preprocess $B$ numbers, such that in constant time, one may retrieve $s_i$ using \textbf{retrieve}($s$, $i$), and retrieve $s_1+\cdots+s_i$ using \textbf{retrieve\_prefix\_sum}($s$, $i$), provided that $s=s_1+\cdots+s_B$. (implemented in Section~\ref{sec_ten_dec})

\bigskip
                                                      
\underline{Query algorithm:}\\[0.05in]
\textbf{rank}($u$)
\vspace{-.05in}\begin{enumerate}
\addtolength{\itemsep}{-4pt}
\item read $s$
\item $\mathtt{DS(A)}$.\textbf{query}($s$, $u$)
\end{enumerate}

$\mathtt{DS(A)}$.\textbf{query}($s$, $u$) \hfill // output $A[0]+\cdots+A[u-1]$ provided that the sum of $A$ is $s$
\vspace{-.05in}\begin{enumerate}
\addtolength{\itemsep}{-4pt}
\item if $u=0$, return 0
\item if $u=|A|$, return $s$ \hfill // must have returned by now, if $|A|=1$
\item compute $(i, u')$ for integer $i$ and $0\leq u'<|A|/B$ such that $u=i(|A|/B)+u'$ \hfill // $u$ is in the $(i+1)$-th block
\item let $s_{i+1}=\mathtt{sums(A)}.\textbf{retrieve}(s,i+1)$
\item return $\mathtt{sums(A)}.\textbf{retrieve\_prefix\_sum}(s,i)$ \hfill // $s_1+\cdots+s_i$ \\
\phantom{return }+\,$\mathtt{DS(A_{i+1})}$.\textbf{query}($s_{i+1}$, $u'$) \hfill // recurse on $A_{i+1}$
\end{enumerate}
}}}
\caption{Succinct rank data structure}\label{fig_datastr}
\end{center}
\end{figure}

\bigskip

Suppose both \textbf{retrieve\_prefix\_sum} and \textbf{retrieve} take constant time, then rank queries can be answered in $O(t)=O(\log_w n)$ time.
Hence, the task boils down to implementing \textbf{aggregate\_sums} efficiently, which constitutes our main technical contribution.

Let $T_i=s_1+\cdots+s_i$.
There are two natural implementations.
\begin{itemize}
	\item We may store the prefix sums $T_1,T_2,T_3,\ldots$, which allows one to retrieve prefix sum and each $s_i$ in constant time.
	However, under uniform input distribution, we have
	\[
		I(T_{k-1};T_k)\approx \frac{1}{2}\log k.
	\]
	``The same $\frac{1}{2}\log k$ bits of information'' will be stored in different locations.
	Since the goal is to use no more than one extra bit, we fail as soon as the first $5$ prefix sums are written down.
	\item Or we may store the numbers $s_1,s_2,s_3,\ldots$ The same issue may still exist. Moreover, storing the numbers explicitly would not let us retrieve prefix sums efficiently.
\end{itemize}
It is worth noting that if we set $B=2$ instead of $w^{1/3}$, one could jointly store the pair $(s_1, s_2)$ \emph{conditioned on} $s_1+s_2$ in one memory word, which introduces no redundancy and allows one to retrieve both sums in constant time.
However, the depth of the tree becomes $\log n$ instead of $\log_w n$, so does the query time.
It is essentially what the previous data structure~\cite{Pat08} does, after ``projected'' into the information cell-probe model.
The major effort of \cite{Pat08} is spent on transforming this ``standard'' solution in information cell-probe model to word RAM, which we will briefly discuss in Section~\ref{sec_making_RAM}.
Hence, the subroutine \textbf{aggregate\_sums} is where our solution deviates from the previous one.

\subsection{Aggregating sums}\label{sec_agg_sum}
To implement these subroutines, the main idea is to store correlated variables $\{T_i\}$ as we discussed in the introduction: Find a random variable $\eta$ such that conditioned on $\eta$ (and the sum of all numbers $T_B=s$), the prefix sums $\{T_i\}$ are approximately independent; we will first store $\eta$, then store the prefix sums \emph{conditioned on $\eta$}.
To retrieve a $T_i$, we first read $\eta$, then read $T_i$, which is stored according to conditional distribution.
Each number $s_i$ can be retrieved by simply taking the difference of $T_i$ and $T_{i-1}$.

To find such $\eta$, we analyze the joint distribution of $(T_1,\ldots,T_B)$, which can be described by an order-$B$ tensor $\bT$ of size $(n+1)^B$.
Each entry $(x_1,\ldots,x_B)$ describes the probability that $s_1=x_1, s_2=x_2,\ldots,s_B=x_B$.
As we mentioned in the introduction, finding such a random variable $\eta$ is (approximately) equivalent to decomposing $\bT$ into a sum of few nonnegative rank-$1$ tensors, since a nonnegative rank-$1$ tensor describes an independent distribution.
Then the joint distribution can be viewed as a convex combination of these independent distributions, where $\eta$ decides which independent distribution to sample from.
The number of such rank-$1$ tensors corresponds to the support size of $\eta$, hence provides an upper bound on its entropy.
Since $\eta$ needs to be stored in one word, the goal is to decompose $\bT$ into $2^{w}=n^{O(1)}$ rank-$1$ nonnegative tensors.
In the next subsection, we elaborate this idea, and show how to find this $\eta$ for our problem.

\subsection{Tensor decomposition}\label{sec_ten_dec}
The tensor corresponding to the joint distribution of $(T_1,\ldots,T_B)$ can be written in an explicit form as follows.
Suppose each subarray $A_1,\ldots,A_B$ has size $l$, we have
\begin{equation}\label{eqn_tensor_T}
	\bT_{x_1,\ldots,x_B}:=\Pr[T_1=x_1\wedge T_2=x_2\wedge \cdots\wedge T_B=x_B]=2^{-lB}\cdot \prod_{i=1}^B \binom{l}{x_i-x_{i-1}},
\end{equation}
where we assumed $x_0=0$.


We first show that the binomial coefficients $\binom{l}{x_i-x_{i-1}}$ appeared in Equation~\eqref{eqn_tensor_T} can be piecewise-approximated by a low degree polynomial in $x_i$ and $x_{i-1}$ (multiplied by an exponential function) (see Lemma~\ref{lem_approx} in Section~\ref{sec_binom_approx}).
After handling the negative terms via Equation~\eqref{eqn_neg} and putting together the piecewise approximation, Lemma~\ref{lem_binom} in Section~\ref{sec_binom_approx} implies that the binomial coefficients can be expressed as\footnote{The lemma provides a few extra guarantees on the approximation that is needed for the word RAM data structure, and is omitted here.}
\[
	\binom{l}{x_i-x_{i-1}}\cdot 2^{-l}=E(x_i,x_{i-1})+\sum_{j=1}^{r_0} Q_j(x_{i-1})\cdot R_j(x_i),
\]
where $E,Q_j,R_j$ are all nonnegative, $\E_{x_{i-1}} [\sum_{x_i} E(x_i, x_{i-1})]\leq \epsilon$ (the error term is small) and $r_0=(\log 1/\epsilon)^{O(1)}$ (the number of terms is small).
By multiplying the above approximation over all $i=1,\ldots,B$, we obtain an approximation for $\bT$:
\[
	\bT_{x_1,\ldots,x_B}=\tilde{E}(x_1,\ldots,x_B)+\sum_{j_1,\ldots,j_B=1}^{r_0}Q_{j_1}(0)R_{j_1}(x_1)Q_{j_2}(x_1)R_{j_2}(x_2)\cdots Q_{j_B}(x_{B-1})R_{j_B}(x_B).
\]
One may verify that the overall error term $\tilde{E}$ is small, $\sum_{x_1,\ldots,x_B}\tilde{E}(x_1,\ldots,x_B)\leq B\epsilon$, and the total number of terms $r$ is at most $r_0^B=(\log 1/\epsilon)^{O(B)}$.
Note that each term $Q_{j_1}(0)R_{j_1}(x_1)Q_{j_2}(x_1)R_{j_2}(x_2)\cdots Q_{j_B}(x_{B-1})R_{j_B}(x_B)$ corresponds to a rank-1 tensor (i.e., an independent distribution).
By normalizing each rank-1 tensor, we obtain the following lemma on nonnegative approximate tensor decomposition of $\bT$.

\begin{lemma}
	One can find $r\leq (\log 1/\epsilon)^{O(B)}$ rank-1 tensors $\bT_1,\ldots,\bT_r$ such that the above tensor $\bT$ can be expressed as
	\[
		\bT=p_E\cdot \bE+\sum_{j=1}^r p_j\cdot \bT_j,
	\]
	where $\bE$, $p_E$, $\bT_j$ and $p_j$ are all nonnegative, $\|\bE\|_1=1,\|\bT_j\|_1=1$ and $p_E\leq B\epsilon$ for all $j=1,\ldots,r$.
\end{lemma}
Note that for technical reasons, the final data structure for word RAM requires extra guarantees on the decomposition, and the above lemma is not directly used in Section~\ref{sec_upper}.
Hence, we only state it here without a formal proof.

By setting $\epsilon=1/n^2$, we have $r=2^{O((\log n)^{1/3}\log\log n)}=n^{o(1)}$.
Thus, one can view the joint distribution $\bT$ of $(T_1,\ldots,T_B)$ as a convex combination of $\bE$, whose probability is tiny, and $r$ mutually independent distributions $\bT_1,\ldots,\bT_r$.\footnote{We abuse the notation of a tensor for the corresponding joint distribution.}
Now let $\eta$ be the random variable indicating the distribution we are currently sampling from, we have $H(\eta)\leq o(\log n)$, and the prefix sums are almost independent conditioned on $\eta$.

To generate $\eta$ given input array $A$, we first partition the sample space $\{0,1\}^{lB}$, such that each part corresponds to one distribution in the convex combination.
More specifically, we fix a partition of $\{0,1\}^{lB}$, the domain of $A$, into $\cK_E, \cK_1,\ldots,\cK_r$, such that $|\cK_E|\approx p_E\cdot 2^{lB}$ and $|\cK_j|\approx p_j\cdot 2^{lB}$ for $j=1,\ldots,r$.
Moreover, this partition guarantees that for every $j=1,\ldots,r$, the prefix sums $(T_1,\ldots,T_B)$ of an array $A$ that is sampled uniformly from $\cK_j$, is approximately distributed according to $\bT_j$.
Also, the distribution of $(T_1,\ldots,T_B)$ when $A$ is sampled uniformly from $\cK_E$ is roughly $\bE$.
Given such a partition, it suffices to set $\eta$ to the part that contains input $A$.
See Figure~\ref{fig_aggr} for the detailed implementations of \textbf{aggregate\_sums}, \textbf{retrieve} and \textbf{retrieve\_prefix\_sum}.

\begin{figure}[!h]
\begin{center}
\fbox{                                                                                                         
{\footnotesize                                                                                                 
\parbox{6.375in} {

Fix a partition $\cK_E,\cK_1,\ldots,\cK_r$ of $\{0,1\}^{lB}$ with the above guarantees.

\vspace{0.15in}

\underline{Aggregating the sums:}\\[0.05in]
\textbf{aggregate\_sums}($A,s_1,\ldots,s_B$)
\vspace{-.05in}\begin{enumerate}
\addtolength{\itemsep}{-4pt}
\item for $i=1,\ldots,B$, compute $T_i=s_1+\cdots+s_i$ 
\item $\eta:=$ the index of the set among $\{\cK_E,\cK_1,\ldots,\cK_r\}$ that $A$ belongs to
\item store $\eta$ conditioned on $T_B$ \hfill // $T_B$, the sum of the entire subarray, is assumed to be stored outside this subroutine\\
\item for $i=1,\ldots,B-1$, store $T_i$ conditioned on $\eta$
\item return the $B$ stored variables $(\,``\eta\mid T_B"\,\,,\,\,\, ``T_1\mid \eta"\,\,,\,\,\,\ldots\,,\,\,\,``T_{B-1}\mid \eta"\,)$
\end{enumerate} 

\bigskip
                                                      
\underline{Query algorithms:}\\[0.05in]
\textbf{retrieve\_prefix\_sum}($s$, $i$) \hfill // output $T_i$ provided that $s_1+\cdots+s_B=s$
\vspace{-.05in}\begin{enumerate}
\addtolength{\itemsep}{-4pt}
\item if $i=0$, return $0$
\item if $i=B$, return $s$
\item read $\eta$ conditioned $s$ \hfill // $T_B=s$
\item read $T_i$ conditioned on $\eta$
\item return $T_i$
\end{enumerate}

\textbf{retrieve}($s$, $i$) \hfill // output $s_i$ provided that $s_1+\cdots+s_B=s$
\vspace{-.05in}\begin{enumerate}
\addtolength{\itemsep}{-4pt}
\item return $\textbf{retrieve\_prefix\_sums}(s, i)\,-\,\textbf{retrieve\_prefix\_sums}(s, i-1)$
\end{enumerate}
}}}
\caption{\textbf{aggregate\_sums} ``implementation''}\label{fig_aggr}
\end{center}
\end{figure}

\bigskip
Intuitively, storing a random variable $\mu$ takes $H(\mu)$ bits of space, and storing $\mu$ conditioned $\nu$ should take $H(\mu\mid \nu)$ bits of space.
We will briefly discuss how to store a variable conditioned on another in the next subsection, and the details are deferred to the final construction in Section~\ref{sec_upper}.

One may verify that $\eta$ generated by this algorithm satisfies
\[
	\sum_{i=1}^B H(T_i\mid \eta)\leq H(T_1,\ldots,T_B\mid \eta)+o(n^{-1}).
\]
Hence, the total space usage of \textbf{aggregate\_sums} is at most
\begin{align*}
	 H(\eta\mid T_B)+\sum_{i=1}^{B-1} H(T_i\mid \eta) &\leq H(\eta\mid T_B)+H(T_1,\ldots,T_B\mid \eta)-H(T_B\mid \eta)+o(n^{-1}) \\
	&=H(T_1,\ldots,T_{B-1}, \eta\mid T_B)+o(n^{-1}).
\end{align*}

Note that the above implementation does not give the right space bound, since the correct space benchmark is $H(T_1,\ldots,T_{B-1}\mid T_B)$ rather than $H(T_1,\ldots,T_{B-1},\eta\mid T_B)$.
In fact, here $\eta$ encodes extra information about the input, which is going to be encoded one more time in the recursion.
This issue is resolved in the final data structure for word RAM by ``adding $\eta$ to the recursion.''
That is, instead of recursing on a subarray provided that its sum has been stored (as what we do here), we will recurse provided that both the sum and ``this extra information encoded by $\eta$'' have already been stored.
Hence, the data structure can avoid storing duplicated information during the recursion.
See Section~\ref{sec_making_RAM} and Section~\ref{sec_upper} for more details.
When no information is stored both in $\eta$ and in the recursion, this subroutine introduces only $o(n^{-1})$ bit of redundancy.
Since it is invoked no more than $O(n)$ times in total, the overall space usage is at most $n+1$ bits.

\subsection{Previous data structure and transforming into RAM}\label{sec_making_RAM}

As mentioned earlier, the main effort in the previous best-known construction is to transform a ``standard'' information cell-probe data structure into word RAM.
To this end, the \emph{spillover representation} was introduced (see Section~\ref{sec_prelim} for its definition), and it is also heavily used in our data structure.

\paragraph{Spillover representation.}
One simple example of the spillover representation is to represent the sum of a 0-1 string.
Given a 0-1 string of length $n$, one can design a data structure using space $[K]\times \{0,1\}^m$, such that $K=n^{O(1)}$, the spillover $k\in [K]$ encodes the sum of all $n$ bits, and $m+\log K\leq n+O(n^{-2})$, i.e., it has $O(n^{-2})$ bit of redundancy.
To do this, we first permute the representation of all $n$-bit binary strings. 
That is, instead of storing the $n$-bit string as is, we are going to sort all $n$-bit strings based on their sums, and store the index in the sorted order as an $n$-bit integer.
The idea is that for most inputs, the top $O(\log n)$ bits of the index already reveals the sum.
Only when the top bits correspond to a boundary between two adjacent sums, the lower bits will need to be read, in order to determine the actual sum.
Let $v=\lceil 3\log n\rceil$.
This problem can be resolved by assigning an integer multiple of $2^{n-v}$ many indices to each sum.
In the other word, all strings with sum equal to $0$ are encoded to the interval $[0, 2^{n-v}-1]$, 
all strings with sum equal to $1$ are encoded to $[2^{n-v}, 2^{n-v+1}-1]$, etc. 
For each sum $x$, all strings with sum equal to $x$ are encoded to an interval of length
\[
	\left\lceil\binom{n}{x}\cdot 2^{v-n}\right\rceil\cdot 2^{n-v},
\]
aligned to integer multiples of $2^{n-v}$.
Therefore, the sum can be retrieved without accessing the lower $n-v$ bits of the encoding.
There are $n+1$ different values for the sum, the largest number needed for this encoding is at most
\[
	2^n+2^{n-v}\cdot (n+1).
\]
Now, we set $m=n-v$ and use the $m$ bits of the memory to store the lower $n-v$ bits, and set the spillover to 
\[
	K=(2^n+2^{n-v}\cdot (n+1))/2^{n-v}=2^v+n+1,
\]
which stores the top bits of the encoding.
The size of the spillover is bounded by $n^{O(1)}$, and the space usage is $m+\log K=\log (2^n+2^{n-v}\cdot (n+1))\leq n+O(n^{-2})$.
This trick is also used later in the formal proof, e.g., see Lemma~\ref{lem_sum_spillover}.

\paragraph{P\v{a}tra\c{s}cu's rank data structure.}
The high-level structure of the previous best-known word RAM data structure is similar to the one stated earlier in this section (e.g., see Figure~\ref{fig_datastr}): a range tree with branching factor $B=2$.
To construct a rank data structure on an array of length $n$, one first recurses on the two halves, and obtains two data structures with spillover. 
It is guaranteed that the two spillover sets are both bounded by $O(n^3)$, and the sum of each half can be decoded from solely the corresponding spillover.
To combine the two data structures, the memory bits obtained from the two recursions are concatenated directly.
The two spillovers can be combined using the above trick for representing the sum, by sorting all pairs of spillovers by the sum of the whole array (this is possible, since each spillover determines the sum of each half).
As we argued above, this step introduces redundancy of $O(n^{-2})$ bit, and the new spillover size is again bounded by $O(n^3)$.
By writing down the final spillover in its binary representation, the overall redundancy is no more than one bit.

\paragraph{Adapt our new data structure to RAM.}
Compared to the previous data structure, our new data structure will use a different algorithm to combine the spillovers (corresponding to \textbf{aggregate\_sums}), and can afford to set the branching factor $B$ to $(\log n)^{1/3}$.
The preprocessing algorithm still recursively constructs a data structure with spillover for each subarray, and provides the same guarantees as in the previous solution: the spillover size is bounded by $\mathrm{poly}\, n$ and the spillover determines the sum.
Then the algorithm combines these $B$ spillovers $k_1,\ldots,k_B$, such that one can retrieve each spillover $k_i$ and compute the sum of first $i$ subarrays in constant time.
Observe that when the space usage is very close to the information theoretical lower bound, for each subarray, the distribution of the sum encoded by a random spillover is close to that of the sum of a random input (Fact~\ref{fact_size_bound}).
That is, for a subarray of length $l$, roughly $\binom{l}{x}\cdot 2^{-l}$ fraction of the spillovers encode the sum equal to $x$.
Therefore, the tensor decomposition argument in the previous subsection also applies when encoding these $B$ spillovers.

More specifically, to compute $\eta$ given the input, we will partition the set of all possible $B$-tuples of spillovers, instead of the set of all inputs, according the tensor decomposition.
We first store $\eta$, and $T_1,\ldots,T_B$ conditioned on $\eta$.
Then, each spillover $k_i$ is stored conditioned on $\eta$ and two adjacent prefix sums $T_{i-1}$ and $T_i$.
We will carefully choose the partition $\{\cK_i\}$ such that all $T_i$ are almost independent conditioned on $\eta$, and moreover, each $k_i$ is almost independent of all other spillovers conditioned on $\eta, T_{i-1}$ and $T_i$.
Hence, for a random input, the space usage is
\begin{align*}
	&\,H(\eta)+\sum_{i=1}^{B}H(T_i\mid \eta)+\sum_{i=1}^B H(k_i\mid \eta,T_{i-1},T_i) \\
	\approx &\,H(\eta)+H(T_1,\ldots,T_B\mid\eta)+\sum_{i=1}^B H(k_i\mid \eta,T_{i-1},T_i) \\
	\approx &\,H(T_1,\ldots,T_B,\eta)+\sum_{i=1}^B H(k_i\mid \eta,k_1,\ldots,k_{i-1},T_{i-1},T_i,k_{i+1},\ldots,k_B) \\
	\leq &\,H(T_1,\ldots,T_B,\eta)+\sum_{i=1}^B H(k_i\mid \eta,k_1,\ldots,k_{i-1},T_1,\ldots,T_B) \\
	= &\,H(T_1,\ldots,T_B,\eta)+H(k_1,\ldots,k_B\mid \eta,T_1,\ldots,T_B) \\
	= &\,H(k_1,\ldots,k_B),
\end{align*}
which approximately matches the information theoretical lower bound of storing $B$ spillovers.

This solution also handles worst-case input.
We will design the partition carefully such that all entropies appeared in the above analysis can be replace by the logarithms of the support sizes.
To encode a variable directly (e.g., $\eta$), we encode the index within its support, which takes logarithm of the support size many bits.
To encode a variable conditioned on another (e.g., $T_i\mid \eta$), the encoding varies based on the value of the conditioning variable.
For instance, to encode $T_i$ given $\eta$, we examine the value of $\eta$, and encode the index of $T_i$ within its support given the value of $\eta$, i.e., all values of $T_i$ that are not compatible with the current value of $\eta$ are removed.
When a support size is a perfect power of two, the corresponding variable can be stored using an integer bits of memory.
Otherwise, Lemma~\ref{lem_change_base} is applied to produce a data structure with spillover, which introduces no more than $n^{-2}$ bit of redundancy each time.
Finally, Lemma~\ref{lem_sum_spillover} ensures that the sum of the whole subarray $T_B$ is encoded in the spillover, to provide the claimed guarantee of the recursion.
See the proof of Lemma~\ref{lem_induct_large} for more details.

As we can see from the calculation above, it also resolves the issue mentioned in the end of Section~\ref{sec_ten_dec}, since after we have encoded the prefix sums, the ``remaining information'' about each spillover is encoded conditioned on both $\eta$ and $s_i$. 
Each spillover $k_i$ is the only parameter that goes into (or comes back from) the next level of recursion, and $k_i$ contains all relevant information about $\eta$ and $s_i$.
Hence, no information revealed by $\eta$ is stored again in the recursion.

\paragraph{Avoid arbitrary word operations.}
The above construction assumes that the query algorithm can perform arbitrary $O(\log n)$-bit word operations.
In P\v{a}tra\c{s}cu's rank data structure, decoding the two spillovers in constant time also requires non-standard operations.
One way to avoid arbitrary word operations is to store a look-up table in the memory.
However, the look-up table itself may take $\mathrm{poly}\, n$ space. 
Here, we apply a standard trick for self-reducible problems (also applied in the previous data structure) to avoid this issue: divide the input array into blocks of length $n^{\delta}$, construct the above data structure and look-up table for each block, and store the sum of first $i$ blocks for all $i$.
To answer a prefix sum query in the $i$-th block, we retrieve the sum of first $i-1$ blocks and make a query in the $i$-th block. 
Since all blocks use the same data structure, the corresponding look-up table is also the same, and only one copy needs to be stored.
By setting $\delta$ to be the right constant, both the size of the look-up table and the total redundancy from all blocks are bounded by $n^{1-c}$ for some constant $c>0$.
See Section~\ref{sec_ram} for details.

\section{Succinct Rank Data Structure}\label{sec_upper}

Guided by the construction in Section~\ref{sec_overview}, in this section, we present a succinct rank data structure that works in the standard word RAM model.
As a starting point, we first present a data structure in the cell-probe model, assuming arbitrary word operations are allowed.
In Section~\ref{sec_ram}, we show how to implement this solution in word RAM.

\begin{restate}[Theorem~\ref{thm_upper_cell}]
	Given a 0-1 array of length $n$ for sufficiently large $n$, for any $t\geq 1$, one can construct a succinct data structure using
	\[
		n+\lceil \frac{n}{w^{\Omega(t)}}\rceil
	\]
	bits of memory in the cell-probe model with word-size $w\geq 7\log n$, such that every rank query can be answered in $O(t)$ time.
\end{restate}
\begin{proof}
	The top-level of the data structure is a sequence of range trees, each with branching factor $B=w^{1/3}$ and depth $t$.
	In the proof, we assume for simplicity that $w$ is an even perfect cube, and $n$ is a multiple of $B^tw$.
	General $n$ and $w$ can be handled via similar approaches.
	Each range tree takes care of a subarray of length $B^t w=w^{t/3+1}$, using $B^t w+1$ bits of memory, i.e., one bit of redundancy.
	We also store the number of ones in the first $i$ subarrays for every $i$ using $\lceil\log (n+1)\rceil$ bits.
	Hence, for every $B^t w$ bits of input, there will be $O(\log n)$ bits of redundancy, which will give us the claimed space bound.

	More specifically, consider a subarray of $B^t w$ bits, and a range tree built from it with branching factor $B$ and depth $t$.
	Each leaf corresponds to a subarray of length $w$.
	Every node at level $t-i$ is the root of a subtree of size $B^i w$ and depth $i$ (assuming root has level $0$).
	We are going to inductively construct data structures for all subtrees.
	The inductive hypothesis is stated below in the claim.

	\begin{claim}\label{cl_induct}
		For each subtree of size $B^i w$ for $i\geq 0$, one can construct a data structure with spillover, using space $[K_i]\times \{0,1\}^{m_i}$, supporting rank queries in $O(i)$ time, where 
		\[
			m_i=B^i w-w
		\]
		and
		\[
			K_i=2^w+(34B)^i\cdot n2^{w/2}.
		\]
		Moreover, the sum of all $B^iw$ bits in the subtree can be computed by only reading the spillover $k_i\in[K_i]$.
	\end{claim}
	Before proving the claim, let us first show that it implies the theorem.
	When $i=t$, 
	\[
		K_t=2^w+(34B)^tn2^{w/2}\leq 2^w+n^{2+o(1)}2^{w/2}\leq 2^{w+1}.
	\]
	Hence, $K_t$ takes $w+1$ bits to store.
	The space usage will be $B^t w+1$ bits.

	Finally, we divide the input into $n'=n/(B^t w)$ subarrays $A_1,A_2,\ldots,A_{n'}$ of length $B^t w$.
	For each subarray, we construct a data structure using Claim~\ref{cl_induct}. 
	We also store the total number of ones in ${A_1\cup A_2\cup\cdots\cup A_i}$ for all $i\in\{1,\ldots,n'-1\}$, each taking $\lceil\log (n+1)\rceil$ bits.
	Note that the sum of all $n'$ subarrays is not necessary to store. 
	When $n'=1$, the redundacy is one bit, otherwise, it is at most $O(n'\log n)$.
	Hence, the total space usage is at most
	\[
		\frac{n}{B^tw}\cdot (B^tw+1)+\left(\frac{n}{B^tw}-1\right)\lceil\log (n+1)\rceil\leq n+\lceil\frac{n}{w^{\Omega(t)}}\rceil
	\]
	as claimed.

	To answer a rank query \rank{u}, we first compute the subarray $A_i$ that $u$ is in.
	By retrieving the number of ones in first $i-1$ subarrays and querying the rank of $u$ within $A_i$, we obtain the answer in $O(t)$ time.
	Hence, it remains to prove the claim (by induction).
	\paragraph{Base case.} The statement is trivial when $i=0$: store the entire subtree, which has only $w$ bits, in the spillover.
	\paragraph{Induction step.} First construct a data structure for each child of the root, which corresponds to a subtree of size $l=B^{i-1}w$, using space $[K_{i-1}]\times \{0,1\}^{m_{i-1}}$ each.
	The key technical part of the induction step is the following lemma that combines $B$ spillovers into one data structure, and allows one to decode each spillover and the sum of first $i$ subtrees in constant time.
	\begin{lemma}\label{lem_induct}
		Given $B$ such spillovers $k_1,k_2,\ldots,k_B\in[2^w+\sigma]$ for $2^{w/2}n\leq \sigma\leq 2^{w}/n$, let $\SUM: [2^w+\sigma]\rightarrow [0, l]$ be the function that decodes the sum from a spillover.
		For $i=0,\ldots,B$, denote by $T_i$ the sum of first $i$ subtrees, i.e., $T_i:=\sum_{j\leq i} \SUM(k_j)$.
		One can construct a data structure using space $[K]\times \{0,1\}^m$ for $K=[2^w+34B\sigma]$ and $m=(B-1)w$, such that for $i=1,\ldots,B$, decoding each $k_i$ and $T_i$ takes constant time, and the spillover determines the sum of the entire subtree $T_B$.
	\end{lemma}

	Its proof is deferred to the next subsection.
	By induction hypothesis, we have
	\[
		n2^{w/2}\leq K_{i-1}-2^w\leq (34B)^tn2^{w/2}\leq n^{2+o(1)}2^{w/2}<2^{w}/n,
	\]
	i.e., $n2^{w/2}<\sigma<2^w/n$.
	Lemma~\ref{lem_induct} lets us combine the $B$ spillovers, and obtain a $[K_i]\times \{0,1\}^{(B-1)w}$-space data structure for 
	\[
		K_i=2^w+(34B)(K_{i-1}-2^w)=2^w+(34B)^i\cdot n2^{w/2}.
	\]
	Hence, in total the data structure uses $B m_{i-1}+(B-1)w=B^i w-w$ bits and a spillover of size $2^w+(34B)^i\cdot n2^{w/2}$, and the spillover determines the sum.

	To answer a rank query \rank{x}, we first compute $i$, the index of the subtree that $x$ is in.
	Then we retrieve $T_{i-1}$ and $k_i$ in constant time by Lemma~\ref{lem_induct}. 
	Given the spillover $k_i$, we may recursively query the rank of $x$ inside the $i$-th subtree.
	The query output can be computed by adding its rank inside the $i$-th subtree to $T_{i-1}$.
	The total query time is proportional to the depth of the tree, which is $O(i)$.
	This proves the theorem.
\end{proof}

\subsection{Combining the spillovers}

The goal of this subsection is to prove Lemma~\ref{lem_induct}.
We first observe that if there is a data structure (with spillover) that allows one to decode the sum of the entire subarray (or subtree) in constant time, then one may assume without loss of generality that the sum is encoded in the spillover, which we state in the following lemma.

\begin{lemma}\label{lem_sum_spillover}
Given input data $Z$, suppose there is a data structure $D$ using space $[K]\times \{0,1\}^m$, which allows one to answer each query $f_i(Z)$ for $i\geq 0$ in time $t_q$, assuming the word-size is $w$.
Then $D$ can be turned into another data structure using space ${[K+r]\times \{0,1\}^m}$, which allows one to answer each query in time $2t_q$, moreover, $f_0(Z)$ can be answered by reading only the spillover, where $r$ is the number of different values that $f_0(Z)$ can take.
\end{lemma}
\begin{proof}
To construct a data structure that stores $f_0(Z)$ in the spillover, we first simulate $D$ on $Z$, which generates $m$ bits of memory $s\in\{0,1\}^m$ and a spillover $k\in [K]$.
Then we simulate the query algorithm for query $f_0(Z)$, which reads $t_q$ words (or $t_qw$ bits) of $s$.
We move those $t_qw$ bits to the spillover, by increasing the spillover size to $K2^{t_qw}$ and removing them from $s$.
The relative order of all other bits in $s$ are unchanged.
This generates a new spillover $k'\in [K2^{t_qw}]$ and a memory $s'$ of $m-t_qw$ bits.
Note that the query algorithm can be adaptive, hence the bits removed from $s$ could vary for different inputs.

Now $k'$ does encode $f_0(Z)$, but its size is much larger than claimed.
To decrease the spillover size back to approximately $K$, observe that we are free to choose any bijection between its domain $[K2^{t_qw}]$ and the pair of original spillover $k$ and the $t_qw$ bits.
Hence, we will pick a representation of $k'$, such that all its values that encode the same value of $f_0(Z)$ are consecutive in $[K2^{t_qw}]$.
For example, we may use a representation such that $f_0(Z)$ is monotone.
Intuitively, if each value of $f_0(Z)$ corresponds to an interval in $[K2^{t_qw}]$, reading the ``top bits'' of $k'$ should likely tell us $f_0(Z)$.
However, if $k'$ lies close to the boundary between two consecutive $f_0(Z)$ values, reading the entire $k'$ may still be required to distinguish between the two.

This issue can be resolved by rounding up the boundaries to integer multiples of $2^{t_qw}$.
That is, we adjust the representation, so that each interval corresponding to a value of $f_0(Z)$ always starts at a multiple of $2^{t_qw}$.
Thus, $f_0(Z)$ can be computed without reading the lowest $t_qw$ bits of $k'$.
Since $f_0(Z)$ can take $r$ different values, this could only increase the spillover set size to at most $(K+r)2^{t_qw}$, which can be viewed as $[K+r]\times \{0,1\}^{t_qw}$
By moving these $t_qw$ bits back to (the beginning of) the memory, we obtain a data structure using space $[K+r]\times \{0,1\}^{m}$, such that $f_0(Z)$ can be answered by reading only the spillover.

To answer a generic query $f_i(Z)$ for $i\geq 1$, one first reads the spillover and first $t_q$ words of the memory.
This determines $k'$, and hence the initial spillover $k$ generated from $D$ as well as all words read by the query algorithm of $D$ when $f_0(Z)$ is queried, which are the words removed from $s$.
In particular, this determines the mapping between words in $s'$ and $s$.
Thus, $f_i(Z)$ can be computed by simulating the query algorithm of $D$.
The total query time is $2t_q$.
\end{proof}
In particular, when $Z$ is a subarray of size $Bl$, and $f_0(Z)$ is the number of ones in it, we have $r=Bl+1\leq n+1$.
Note that this lemma is applied once at each level of the recursion, thus the total query time would at most increase by a factor of two.
To prove Lemma~\ref{lem_induct}, we will use different constructions based on the value of $l$, the length of each subarray.
The two cases are stated below in Lemma~\ref{lem_induct_small} and Lemma~\ref{lem_induct_large} respectively.
Since $n+1\leq B\sigma$, Lemma~\ref{lem_induct} is an immediate corollary of Lemma~\ref{lem_sum_spillover}, \ref{lem_induct_small} and \ref{lem_induct_large}.

\begin{lemma}\label{lem_induct_small}
	If $B\log (l+1)\leq w/2$, given $k_1,\ldots,k_B\in[2^w+\sigma]$ for $n2^{w/2}\leq\sigma\leq 2^w/n$, one can construct a data structure using space $[K]\times \{0,1\}^{m}$ for $m=(B-1)w$ and $K\leq 2^w+2B\sigma$, such that each $k_i$ and $T_i$ can be decoded in constant time.
\end{lemma}

\begin{lemma}\label{lem_induct_large}
	If $B\log (l+1)> w/2$, given $k_1,\ldots,k_B\in[2^w+\sigma]$ for $n2^{w/2}\leq\sigma\leq 2^w/n$, one can construct a data structure using space $[K]\times \{0,1\}^{m}$ for $m=(B-1)w$ and $K\leq 2^w+33B\sigma$, such that each $k_i$ and $T_i$ can be decoded in constant time.
\end{lemma}

Recall that each spillover $k_i\in[2^w+\sigma]$ together with $l-w$ additional bits encodes a subarray of length $l$, and $k_i$ encodes the number of ones in this subarray, which is decoded by the function $\SUM$.
Since the space usage is close to the information theoretical limit, if we sample a random $k_i$, the distribution of $\SUM(k_i)$ should be close to the distribution of the sum of the subarray, i.e., the binomial distribution $B(l, 1/2)$.
In particular, we have the following facts by counting.

\begin{fact}\label{fact_size_bound}
	For every $x\in[0, l]$, we have
	\[
		|\SUM^{-1}(x)|\geq \binom{l}{x}\cdot 2^{-l+w}.
	\]
	For any subset $X\subseteq [0, l]$, 
	\[
		|\SUM^{-1}(X)|\leq \sigma+\sum_{x\in X}\binom{l}{x}\cdot 2^{-l+w}.
	\]
\end{fact}

\begin{proof}
	The encoding supports \rank{u} operations for all $0\leq u\leq l$, and the answers to all queries recover the entire subarray.
	Hence, all $2^l$ different subarrays of length $l$ must have different encodings.
	On the other hand, for every $x\in[0, l]$, the number of subarrays with sum equal to $x$ is $\binom{l}{x}$, and each $k_i$ corresponds to only $2^{l-w}$ different encodings.
	Therefore, at least $\binom{l}{x}\cdot 2^{-l+w}$ different $k_i$ should encode arrays with the sum equal to $x$, i.e.,
	\[
		|\SUM^{-1}(x)|\geq \binom{l}{x}\cdot 2^{-l+w}.
	\]

	By subtracting the complement of $X$ from the universe, we have
	\begin{align*}
		|\SUM^{-1}(X)|&= 2^w+\sigma-\sum_{x\notin X}|\SUM^{-1}(x)|\\
		&\leq 2^w+\sigma-\sum_{x\notin X}\binom{l}{x}\cdot 2^{-l+w}\\
		&= \sigma+\sum_{x\in X}\binom{l}{x}\cdot 2^{-l+w}.
	\end{align*}
\end{proof}

One crucial subroutine used in several parts of our construction is a succinct data structure storing ``uniform and independent'' elements with nearly no redundancy from \cite{DPT10}, which we state in the following lemma.
\begin{lemma}\label{lem_change_base}
	Suppose we are given a $B$-tuple $(x_1,x_2,\ldots,x_B)\in [M_1]\times [M_2]\times\cdots\times [M_B]$, such that $B\ll 2^{w/2}$ and $M_i\leq 2^w$ for every $i\in [B]$.
	Then for every integer $m\geq \sum_{i=1}^B\log M_i-w$, there is a data structure that uses space $[K]\times \{0,1\}^m$ for ${K}=\left\lceil 2^{-m}\cdot\prod_{i=1}^B M_i\right\rceil+1$, and allows one to decode each $x_i$ in constant time.
\end{lemma}
The original theorem in~\cite{DPT10} is stated only for numbers from the same domain, i.e., $M_1=\cdots=M_B$.
However, the same idea also applies when the domains are different.
The proof of Lemma~\ref{lem_change_base} can be found in Appendix~\ref{app_change_base}.

\bigskip

We first present the construction for small $l$, which is similar to~\cite{Pat08}.

\begin{proof}[Proof of Lemma~\ref{lem_induct_small}]
	When $B\log (l+1)\leq w/2$, the sums of $B$ subarrays can all fit in one spillover.
	We use the spillover to store the sums, then encode the subarrays \emph{conditioned on} the sums.

	More specifically, let $m=(B-1)w$.
	For every $B$-tuple of sums $\bs=(s_1,\ldots,s_B)\in [0, l]^B$, by Fact~\ref{fact_size_bound}, we have
	\begin{align*}
		\sum_{i=1}^B\log \left|\SUM^{-1}(s_i)\right|-w&\leq \sum_{i=1}^B\log \left(\sigma+\binom{l}{s_i}\cdot 2^{-l+w}\right)-w \\
		&\leq \sum_{i=1}^B\log \left(2^{w-1}+2^w/\sqrt{l}\right)-w \\
		&\leq Bw-w \\
		&\leq m.
	\end{align*}
	By Lemma~\ref{lem_change_base}, for every $\bs$, there is a data structure encoding a tuple $(k_1,\ldots,k_B)$ such that $\SUM(k_i)=s_i$, using space $[K_{\bs}]\times \{0,1\}^{m}$, where
	\[
		K_{\bs}\leq 2^{-Bw+w}\cdot \prod_{i=1}^B|\SUM^{-1}(s_i)|+2.
	\]

	We then ``glue together'' these $(l+1)^B$ data structures for different tuples of sums by taking the union of the spillover sets.
	That is, let $K=\sum_{\bs} K_{\bs}$, we can view $[K]$ as the set of pairs $\{(\bs,k):\bs\in [0,l]^B, k\in [K_{\bs}]\}$ (via a fixed bijection hard-wired in the data structure).
	Given an input $(k_1,\ldots,k_B)$, we first compute the sums $\bs=(s_1,\ldots,s_B)$, and encode the input using the above data structure for $\bs$, which generates $m$ bits and a spillover $k\in [K_{\bs}]$.
	The final data structure will consist of these $m$ bits and the spillover $(\bs, k)$, encoded in $[K]$.
	The size of the spillover set is at most
	\begin{align*}
		K&=\sum_{\bs\in [0, l]^B} K_{\bs} \\
		&\leq \sum_{\bs\in [0, l]^B}(2^{-Bw+w}\cdot \prod_{i=1}^B |\SUM^{-1}(s_i)|+2) \\
		&= 2^{-Bw+w}\cdot \prod_{i=1}^B \sum_{s_i=0}^l |\SUM^{-1}(s_i)| + 2(l+1)^B \\
		&\leq 2^{-Bw+w}\cdot (2^w+\sigma)^B+2^{w/2+1} \\
		&= 2^{w}\cdot (1+\sigma 2^{-w})^B+2^{w/2+1}, \\
		\intertext{which by the bounds on $\sigma$ that $n2^{w/2}<\sigma<2^{w}/n$, is at most}
		&\leq 2^w+2B\sigma.
	\end{align*}

	Decoding $T_i$ or $k_i$ can be done in constant time by a straightforward algorithm:
	First decode the pair $(\bs, k)$, which already determines the value of $T_i(=s_1+\cdots+s_i)$, $k_i$ can then be decoded using the decoding algorithm for tuple $\bs$ from Lemma~\ref{lem_change_base}.
\end{proof}

When $l$ is large, the key step is to find a random variable $\eta$ such that all $T_i$ are uniform and independent conditioned on (most values of) $\eta$. 
This allows us to first encode $\eta$, then encode the prefix sums $\{T_i\}$ nearly optimally using Lemma~\ref{lem_change_base}.
Finally, we encode the spillovers $\{k_i\}$ conditioned on the prefix sums.

In the following lemma, we first present a ``not-so-efficient'' solution, which will be used as a subroutine in our final construction.


\begin{lemma}\label{lem_induct_lowprob}
	For any $B'\leq B$, given a sequence $(k_1,\ldots,k_{B'})$, one can construct a data structure using $B'(w+\log w)$ bits of space, such that each $k_i$ and $T_i$ can be retrieved in constant time.
\end{lemma}
\begin{proof}

	We first partition $[2^w+\sigma]$, the domain of each $k_i$, into two sets based on the sum it encodes: $\cK_{\textrm{high}}:=\SUM^{-1}(l/2\pm\sqrt{lw})$ and $\cK_{\textrm{low}}:=\SUM^{-1}([0, l]\setminus (l/2\pm\sqrt{lw}))$.
	The idea is that if several consecutive spillovers are in $\cK_{\textrm{high}}$, then it takes few bits to encode their sum; if a spillover is in $\cK_{\textrm{low}}$, then it takes few bits to encode the spillover itself.
	By Fact~\ref{fact_size_bound}, 
	\begin{align*}
		|\cK_{\textrm{low}}|&=|\SUM^{-1}([0, l]\setminus (l/2\pm\sqrt{lw}))| \\
		&\leq \sigma+\sum_{x\notin l/2\pm\sqrt{lw}} \binom{l}{x}2^{-l+w} \\
		&\leq \sigma+2^{w+1-2w} \\
		&\leq \sigma+1.
	\end{align*}

	The first $B'$ bits encode for each $i\in[B']$, if $k_i$ is in $\cK_{\textrm{high}}$ or in $\cK_{\textrm{low}}$.
	Then we allocate $w+\lfloor\log w\rfloor-1$ consecutive bits to each $k_i$, where we store extra information about each $k_i$ as follows.

	If $k_i\in\cK_{\textrm{low}}$, we spend $\lceil\log (Bl+1)\rceil$ bits to write down $T_i$, the sum of first $i$ blocks, and $\lceil\log (\sigma+1)\rceil$ bits to encode $k_i$ within $\cK_{\textrm{low}}$.
	The space usage is at most
	\[
		\log Bl+\log \sigma+2\leq w+2<w+\lfloor \log w\rfloor -1
	\]
	bits, since $\sigma<2^{w}/n$ and $Bl\leq n$.

	If $k_i\in\cK_{\textrm{high}}$, denote by $i_{\textrm{pred}}$ the closest block preceding $i$ that is not in $\cK_{\textrm{high}}$, i.e., $$i_{\textrm{pred}}:=\max\{j:j<i,k_j\in\cK_{\textrm{low}}\textrm{ or }j=0\}.$$
	We first spend $\lceil\log (2B\lfloor \sqrt{lw}\rfloor +1)\rceil$ bits to encode $T_i-T_{i_{\textrm{pred}}}$.
	This is possible since all subarrays in-between have their sums in a consecutive range of length $2\lfloor \sqrt{lw}\rfloor$.
	Then we spend another $\lceil\log (\sigma+2^w/\sqrt{l})\rceil$ bits to encode $k_i$ conditioned on $\SUM(k_i)$.
	Again such encoding is possible, since for any $x$, by Fact~\ref{fact_size_bound},
	\begin{align*}
		\left|\SUM^{-1}(x)\right|&\leq \sigma+\binom{l}{x}\cdot 2^{-l+w} \\
		&\leq \sigma+\binom{l}{\lfloor l/2\rfloor}\cdot 2^{-l+w} \\
		&\leq \sigma+2^w/\sqrt{l}.
	\end{align*}

	In this case, the space usage is at most
	\begin{align*}
		\log (2B\sqrt{lw})+\log(\sigma+2^w/\sqrt{l})+2&=\log (2B\sqrt{lw}\sigma+2^{w+1}B\sqrt{w})+2 \\
		&\leq \log (2^{w+1}\sqrt{w/l}+2^{w+1}B\sqrt{w})+2 \\
		&=w+1+\log B+\frac{1}{2}\log w+\log ((B\sqrt{l})^{-1}+1)+2 \\
		&\leq w+\log w-1.
	\end{align*}
	where the first inequality uses $\sigma<2^w/n$ and $Bl\leq n$, and the last inequality uses $B=w^{1/3}$.

	The total space usage is at most $B'(w+\log w)$.
	To decode $T_i$, one reads the first $B'$ bits in constant time (as $B'<w$) to retrieve for \emph{every} $j$, whether $k_j\in \cK_{\textrm{high}}$ or $\cK_{\textrm{low}}$.
	If $k_i\in \cK_{\textrm{low}}$, we have explicitly stored the value of $T_i$.
	Retrieving its value thus takes constant time.
	If $k_i\in \cK_{\textrm{high}}$, one first computes $i_{\textrm{pred}}$ using the $B'$ bits retrieved earlier, and reads $T_i-T_{i_{\textrm{pred}}}$.
	It reduces the problem to decoding $T_{i_{\textrm{pred}}}$.
	If $i_{\textrm{pred}}=0$, the problem is solved.
	Otherwise, $k_{i_{\textrm{pred}}}\in \cK_{\textrm{low}}$, and one may apply the above query algorithm.
	In all cases, $T_i$ can be decoded in constant time.

	To decode $k_i$, if $k_i\in \cK_{\textrm{low}}$, $k_i$ is also explicitly encoded.
	Otherwise, $k_i\in \cK_{\textrm{high}}$, and one applies the above query algorithm to retrieve both $T_i$ and $T_{i-1}$, thus determines the value of $\SUM(k_i)$ by taking their difference.
	The value of $k_i$ conditioned on $\SUM(k_i)$ is encoded in the data structure, and can therefore be retrieved in constant time.

	See Figure~\ref{fig_lowprob} for the construction pictorially.
	\begin{figure}[ht]
		\begin{center}
		\begin{tikzpicture}
			\node[rectangle, draw, minimum width=150pt, inner sep=0, minimum height=15pt] at (0,0) {\scriptsize for each $i$, if $k_i\in \cK_{\textrm{high}}$ or $\cK_{\textrm{low}}$};
			\draw [decoration={brace}, decorate] (-75pt, 10pt) -- (75pt, 10pt) node [above=3pt, pos=0.5] {\scriptsize $B'$ bits};

			\node[anchor=west, text width=95pt] at (95pt, -15pt) {\scriptsize ($k_1\in\cK_{\textrm{low}}$)};
			\node[rectangle,draw, minimum width=40pt, inner sep=0, minimum height=15pt, anchor=west] at (75pt, 0) {\scriptsize $k_1$};
			\node[rectangle,draw, minimum width=60pt, inner sep=0, minimum height=15pt, anchor=west] at (115pt, 0) {\scriptsize $T_1$};
			\draw [decoration={brace}, decorate] (75pt, 10pt) -- (175pt, 10pt) node [above=3pt, pos=0.5] {\scriptsize $w+\lfloor\log w\rfloor-1$ bits};

			\node[anchor=west, text width=95pt] at (195pt, -15pt) {\scriptsize ($k_2\in\cK_{\textrm{high}}$)};
			\node[rectangle,draw, minimum width=40pt, inner sep=0, minimum height=15pt, anchor=west] at (175pt, 0) {\scriptsize $k_2$};
			\node[rectangle,draw, minimum width=60pt, inner sep=0, minimum height=15pt, anchor=west] at (215pt, 0) {\scriptsize $T_2-T_1$};
			\draw [decoration={brace}, decorate] (175pt, 10pt) -- (275pt, 10pt) node [above=3pt, pos=0.5] {\scriptsize $w+\lfloor\log w\rfloor-1$ bits};

			\node[anchor=west, text width=95pt] at (295pt, -15pt) {\scriptsize ($k_3\in\cK_{\textrm{high}}$)};
			\node[rectangle,draw, minimum width=40pt, inner sep=0, minimum height=15pt, anchor=west] at (275pt, 0) {\scriptsize $k_3$};
			\node[rectangle,draw, minimum width=60pt, inner sep=0, minimum height=15pt, anchor=west] at (315pt, 0) {\scriptsize $T_3-T_1$};
			\draw [decoration={brace}, decorate] (275pt, 10pt) -- (375pt, 10pt) node [above=3pt, pos=0.5] {\scriptsize $w+\lfloor\log w\rfloor-1$ bits};
		\end{tikzpicture}
		\end{center}
		\caption{memory content for $B'=3$}\label{fig_lowprob}
	\end{figure}
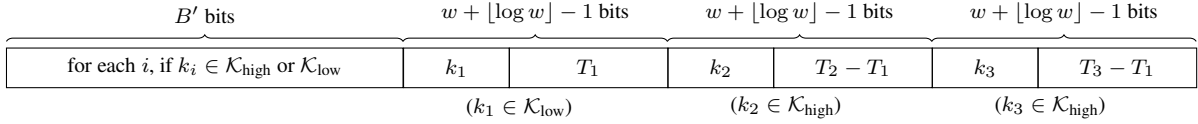
\end{proof}

To prove Lemma~\ref{lem_induct_large}, we will need the following lemma for approximating binomial coefficients, whose proof is presented in the next subsection.

\begin{lemma}\label{lem_part}
	For any large even integer $l$, positive numbers $M_x, M_y$ and $\epsilon$, such that $l>8M_x$, $l>8M_y$ and $\epsilon>2^{-C\sqrt{l}/2+8}$, we have
	\[
		\binom{l}{l/2 +x+y}\cdot 2^{-l+w}=E(x, y)+\sum_{i=1}^r 2^{e_i} \bOne_{X_i}(x)\bOne_{Y_i}(y)
	\]
	for all integers $x\in [-M_x, M_x]$ and $y\in [-M_y, M_y]$, such that 
	\begin{enumerate}[a)]
		\item $E(x, y)\geq 0$ and for every $x\in[-M_x, M_x]$, $\sum_{y=-M_y}^{M_y} E(x, y)\leq \epsilon 2^w+4rM_y2^{w/2}$;
		\item for every $i\in [r]$, $e_i\geq 0$ is an integer, $X_i\subseteq [-M_x, M_x]$, $Y_i\subseteq [-M_y, M_y]$ are sets of integers;
		\item $r\leq O((M_xM_y/l)w^2\log^4 (1/\epsilon))$.
	\end{enumerate}
\end{lemma}

Using the above two lemmas, we are ready to prove Lemma~\ref{lem_induct_large}.

\begin{proof}[Proof of Lemma~\ref{lem_induct_large}]
	By Fact~\ref{fact_size_bound}, for each $i\in [B]$, $\SUM(k_i)=T_{i+1}-T_i$ is distributed approximately according to the binomial distribution $B(l, 1/2)$ for a random $k_i$. 
	We first apply Lemma~\ref{lem_part} to approximate the probability masses of $B(l, 1/2)$, which are binomial coefficients, hence approximating $|\SUM^{-1}(x)|$ for $x\in [0, n]$.

	For each $i\in [B]$, let $\epsilon=\sigma\cdot 2^{-w-2}$, $M_x=(i-1)\sqrt{l\log 1/\epsilon}$, $M_y=i\sqrt{l\log 1/\epsilon}$.
	Since $B=w^{1/3}$ and $B\log (l+1)\geq w/2$, we have
	\[
		l/M_y\geq \sqrt{l/(B^2\log 1/\epsilon)}\geq \sqrt{l}/w>8,
	\]
	similarly $l>8M_x$,
	and
	\[
		\epsilon=2^{-w-2+\log \sigma} \geq 2^{-w-2} \geq 2^{-O(\log^2 l)} \geq 2^{-o(\sqrt{l})}.
	\]

	Hence, by Lemma~\ref{lem_part} (setting $x=(i-1)l/2-T_{i-1}$ and $y=T_i-il/2$), there exists $E_i$, $X_{i, j}$ and $Y_{i, j}$ such that for all $T_{i-1}\in (i-1)(l/2\pm \sqrt{l\log 1/\epsilon})$ and $T_i\in i(l/2\pm \sqrt{l\log 1/\epsilon})$,
	\begin{equation}\label{eqn_part}
		\binom{l}{T_i-T_{i-1}}\cdot 2^{-l+w} = E_i(T_{i-1},T_{i})+\sum_{j=1}^{r_i} 2^{e_{i,j}} \bOne_{X_{i,j}}(T_{i-1})\bOne_{Y_{i,j}}(T_{i})
	\end{equation}
	for integers $e_{i,j}\geq 0$ and 
	$$r_i\leq O((M_xM_y/l)w^2\log^4 (1/\epsilon))\leq O(B^2w^2\log^5{1/\epsilon})\leq w^{O(1)}.$$

	Since $\SUM(k_i)$ approximately follows a binomial distribution, we can partition its domain according to Equation~\eqref{eqn_part}, as follows.

	\begin{claim}\label{cl_part}
		For every $i\in [B]$ and $T_{i-1}\in (i-1)(l/2\pm\sqrt{l\log 1/\epsilon})$, there exists a partition of $[2^w+\sigma]$ (domain of $k_i$) into $\{\cK^{(T_{i-1})}_{i, j}\}_{0\leq j\leq r_i}$, such that
			\begin{enumerate}[(a)]
				\item\label{enum_part_a} $|\cK^{(T_{i-1})}_{i,0}|\leq 2\sigma$ and for all $k_i\in [2^w+\sigma]$ such that $\SUM(k_i)\notin l/2\pm \sqrt{l\log 1/\epsilon}$, we have $k_i\in \cK^{(T_{i-1})}_{i,0}$;
				\item\label{enum_part_b} for $j=1,\ldots,r_i$ and $T_i\in T_{i-1}+l/2\pm \sqrt{l\log 1/\epsilon}$,
				\[
					\left|\cK^{(T_{i-1})}_{i,j}\cap \SUM^{-1}(T_i-T_{i-1})\right|=2^{e_{i,j}}\bOne_{X_{i,j}}(T_{i-1})\bOne_{Y_{i,j}}(T_{i}).
				\]
			\end{enumerate}
	\end{claim}

	To focus on the construction of our data structure, we deferred its proof to the end of the subsection.
	Now let us fix one such partition $\{\cK^{(T_{i-1})}_{i, j}\}_{0\leq j\leq r_i}$ for every $i\in [B]$ and $T_{i-1}\in (i-1)(l/2\pm\sqrt{l\log 1/\epsilon})$.

	\begin{definition}
	For $B'\leq B$, a sequence of $B'$ spillovers $(k_1,\ldots,k_{B'})$ is \emph{good}, if for every $i\in [B']$, $k_i\notin \cK^{(T_{i-1})}_{i,0}$, where $T_i=\sum_{j\leq i} \SUM(k_j)$.
	\end{definition}

	Note that $\cK^{(T_{i-1})}_{i,0}$ is only defined when $T_{i-1}\in (i-1)(l/2\pm \sqrt{l\log 1/\epsilon})$.
	However, if $k_{i-1}\notin \cK^{(T_{i-2})}_{i-1,0}$, then by Item (\ref{enum_part_a}) above, $\SUM(k_{i-1})\in l/2\pm \sqrt{l\log 1/\epsilon}$.
	Hence, if $T_{i-2}\in(i-2)(l/2\pm\sqrt{l\log 1/\epsilon})$, we must also have $T_{i-1}\in (i-1)(l/2\pm \sqrt{l\log 1/\epsilon})$, and thus ``good sequences'' are well-defined.

	\bigskip

	Now we are ready to describe our construction for large $l$.
	We first handle good sequences.

	\paragraph{Input sequence $(k_1,\ldots,k_B)$ is good.}
	Given a good sequence $(k_1,\ldots,k_B)$, one can compute $(T_0,\ldots,T_B)$, and for each $i$, the index of set $j_i$ which $k_i$ is in according to the partition, i.e., $j_i\in \{1,\ldots,r_i\}$ such that $k_i\in \cK^{(T_{i-1})}_{i,j_i}$.
	We first construct a data structure \emph{given} the sequence of indices $\bj=(j_1,\ldots,j_B)$.
	\begin{claim}\label{cl_good_index}
		For every $\bj=(j_1,\ldots,j_B)$ such that $j_i\in \{1,\ldots,r_i\}$ for all $i\in [B]$, given a sequence of spillovers $(k_1,\ldots,k_B)$ such that $k_i\in \cK^{(T_{i-1})}_{i,j_i}$, one can construct a data structure using space $[K_{\bj}]\times \{0,1\}^{(B-1)w}$ for $$K_{\bj}=\left\lceil2^{-(B-1)w}\cdot \prod_{i=1}^B \left(2^{e_{i,j_i}}\left|X_{i+1,j_{i+1}}\cap Y_{i,j_i}\right|\right)\right\rceil+1,$$
		which allows one to decode each $k_i$ and $T_i$ in constant time.\footnote{$X_{B+1,j_{B+1}}$ is assumed to be the entire domain $[2^w+\sigma]$.}
	\end{claim}
	To construct such a data structure, we are going to encode each $k_i$ within $\cK^{(T_{i-1})}_{i,j_i}\cap \SUM^{-1}(T_i-T_{i-1})$, and encode each $T_i$ within $X_{i+1,j_{i+1}}\cap Y_{i,j_i}$ using Lemma~\ref{lem_change_base}.

	More specifically, for every $i$, we know that $k_i\in\cK^{(T_{i-1})}_{i,j_i}$ and $\SUM(k_i)=T_i-T_{i-1}$.
	One can spend $e_{i,j_i}$ bits to encode the index of $k_i$ within $\cK^{(T_{i-1})}_{i,j_i}\cap \SUM^{-1}(T_i-T_{i-1})$, which has size at most $2^{e_{i,j_i}}$ by Item~(\ref{enum_part_b}) in Claim~\ref{cl_part} (in fact, it will be exactly $2^{e_{i,j_i}}$).
	Note that the encoding length of this part does not depend on the input.

	We also know that by Item~(\ref{enum_part_b}), no input sequence will have $T_{i-1}\notin X_{i,j_i}$ or $T_i\notin Y_{i,j_i}$ for any $i\in [B]$.
	That is, we must have $T_i\in X_{i+1,j_{i+1}}\cap Y_{i,j_i}$.
	One can thus apply Lemma~\ref{lem_change_base} to encode each $T_i$ within the set $X_{i+1,j_{i+1}}\cap Y_{i,j_i}$, for $m=(B-1)w-\sum_{i=1}^B e_{i,j_i}$.
	The premise of the lemma is satisfied, since
	\begin{align*}
		m-(\sum_{i=1}^B\log \left|X_{i+1,j_{i+1}}\cap Y_{i,j_i}\right|-w)&=Bw-\sum_{i=1}^B\left(e_{i,j_i}+\log \left|X_{i+1,j_{i+1}}\cap Y_{i,j_i}\right|\right) \\
		&\geq \sum_{i=1}^B\left(w-e_{i,j_i}-\log \left|Y_{i,j_i}\right|\right) \\
		&\geq 0.
	\end{align*}
	The last inequality is due to Equation~\eqref{eqn_part}.
	Hence, Lemma~\ref{lem_change_base} constructs a data structure with spillover size 
	\begin{align*}
		K_{\bj}&=\left\lceil2^{-m}\cdot \prod_{i=1}^B \left|X_{i+1,j_{i+1}}\cap Y_{i,j_i}\right|\right\rceil+1 \\
		&=\left\lceil2^{-(B-1)w}\cdot \prod_{i=1}^B \left(2^{e_{i,j_i}}\left|X_{i+1,j_{i+1}}\cap Y_{i,j_i}\right|\right)\right\rceil+1.
	\end{align*}
	The total space usage is $(B-1)w$ bits with a spillover of size $K_{\bj}$.

	To decode a $T_i$, one can simply invoke the decoding algorithm from Lemma~\ref{lem_change_base}, since $\bj$ is given and all sets ${X_{i+1,j_{i+1}}\cap Y_{i, j_i}}$ are known.
	To decode a $k_i$, one first decodes $T_{i-1}$ and $T_i$, after which both sets $\cK^{(T_{i-1})}_{i,j_i}$ and $\SUM^{-1}(T_i-T_{i-1})$ are known.
	Then $k_i$ can be decoded by retrieving its index within $\cK^{(T_{i-1})}_{i,j_i}\cap \SUM^{-1}(T_i-T_{i-1})$.
	
	\bigskip
	To obtain a data structure for all good sequences, we ``glue'' the above data structures for all $\bj$ in a similar way to Lemma~\ref{lem_induct_small}.
	Let $K_{\textrm{good}}=\sum_{\bj} K_{\bj}$.
	We may view the set $[K_{\textrm{good}}]$ as $\{(\bj, k):k\in[K_{\bj}]\}$ (via a fixed bijection hard-wired in the data structure).
	Given a good sequence $k_1,\ldots,k_B$, one first computes $\bj=(j_1,\ldots,j_B)$ and constructs a data structure using Claim~\ref{cl_good_index}, which generates $(B-1)w$ bits and a spillover $k\in[K_{\bj}]$.
	The data structure will consist of these $(B-1)w$ bits and a final spillover of pair $(\bj,k)$, encoded in $[K_{\textrm{good}}]$.
	To decode $T_i$ or $k_i$, it suffices to decode the pair $(\bj, k)$, and then invoke the decoding algorithm from Claim~\ref{cl_good_index}.

	The spillover size is 
	\begin{align*}
		K_{\textrm{good}}&=\sum_{\bj\in [r_1]\times \cdots\times [r_B]} K_{\bj} \\
		&\leq \sum_{\bj\in [r_1]\times \cdots\times [r_B]} \left(2^{-(B-1)w}\cdot \prod_{i=1}^B \left(2^{e_{i,j_i}}\left|X_{i+1,j_{i+1}}\cap Y_{i,j_i}\right|\right)+2\right) \\
		&=2^{-(B-1)w}\cdot \sum_{\bj\in [r_1]\times \cdots\times [r_B]}\prod_{i=1}^B\left(2^{e_{i,j_i}}\cdot \sum_{T_i\in il/2\pm i\sqrt{l\log 1/\epsilon}}\bOne_{X_{i+1,j_{i+1}}}(T_i)\bOne_{Y_{i,j_i}}(T_i)\right)+w^{O(B)} \\
		&=2^{-(B-1)w}\cdot \sum_{\bj\in [r_1]\times \cdots\times [r_B]}\sum_{\stackrel{T_0,T_1,\ldots,T_B:}{T_i\in il/2\pm i\sqrt{l\log 1/\epsilon}}}\prod_{i=1}^B\left(2^{e_{i,j_i}}\cdot\bOne_{X_{i+1,j_{i+1}}}(T_{i})\bOne_{Y_{i,j_i}}(T_i)\right)+w^{O(B)}, \\
		\intertext{which by the fact that $\bOne_{X_{B+1,i_{B+1}}}(T_B)=\bOne_{X_{1,i_1}}(T_0)=1$, is equal to}
		&=2^{-(B-1)w}\cdot \sum_{\stackrel{T_0,T_1,\ldots,T_B:}{T_i\in il/2\pm i\sqrt{l\log 1/\epsilon}}}\sum_{\bj\in [r_1]\times \cdots\times [r_B]}\prod_{i=1}^B\left(2^{e_{i,j_i}}\cdot\bOne_{X_{i,j_{i}}}(T_{i-1})\bOne_{Y_{i,j_i}}(T_i)\right)+w^{O(B)} \\
		&=2^{-(B-1)w}\cdot \sum_{\stackrel{T_0,T_1,\ldots,T_B:}{T_i\in il/2\pm i\sqrt{l\log 1/\epsilon}}}\prod_{i=1}^B\sum_{j_i=1}^{r_i}2^{e_{i,j_i}}\cdot\bOne_{X_{i,j_{i}}}(T_{i-1})\bOne_{Y_{i,j_i}}(T_i)+w^{O(B)}, \\
		\intertext{which by Equation~\eqref{eqn_part}, is at most}
		&\leq 2^{-(B-1)w}\cdot \sum_{T_0,T_1,\ldots,T_B:T_0=0}\prod_{i=1}^B\left(\binom{l}{T_i-T_{i-1}}\cdot 2^{-l+w}\right)+w^{O(B)} \\
		&= 2^w\cdot \sum_{T_0,T_1,\ldots,T_B:T_0=0}\prod_{i=1}^{B}\left(\binom{l}{T_i-T_{i-1}}\cdot 2^{-l}\right)+w^{O(B)} \\
		&=2^w+w^{O(B)}.
	\end{align*}
	Hence, one can construct a data structure for good sequence using space $[K_{\textrm{good}}]\times \{0,1\}^{(B-1)w}$ for 
	\begin{equation}\label{eqn_k_good}
		K_{\textrm{good}}=2^w+w^{O(B)},
	\end{equation}
	such that each $k_i$ and $T_i$ can be decoded in constant time.
	In particular, one may also choose to use $Bw+1$ bits in total, by rounding up the spillover to $w+1$ bits.
	One can verify that the above construction also applies to any shorter good sequence of length $B'\leq B$, using $B'w+1$ bits of space, which will be used as a separate subroutine below.

	\bigskip

	When the input sequence is not good, there is a smallest $i^*$ such that $k_{i^*}\in \cK^{(T_{i^*-1})}_{i^*,0}$.
	We are going to use different constructions based on whether the suffix $(k_{i^*+1},\ldots,k_B)$ is good, i.e., whether we have for every $i=i^*+1,\ldots,B$, $k_{i}\notin \cK^{(T_{i-1}-T_{i^*})}_{i-i^*,0}$.
	If $(k_{i^*+1},\ldots,k_B)$ is good, we apply the above construction for good sequences to both prefix and suffix.
	The details are presented below.

	\paragraph{$i^*$ breaks the input into two good subsequences.} Suppose the input has one $i^*\in [B]$ such that 
	\begin{itemize}
		\item $k_{i^*}\in \cK^{(T_{i^*-1})}_{i^*,0}$;
		\item $(k_1,\ldots,k_{i^*-1})$ is good;
		\item $(k_{i^*+1},\ldots,k_B)$ is good.
	\end{itemize}
	Since $|\cK^{(T_{i^*-1})}_{i^*,0}|\leq 2\sigma$, one can spend $\lceil\log B\rceil+\lceil\log 2\sigma\rceil$ bits to encode $i^*$ and the index of $k_{i^*}$ within $\cK^{(T_{i^*-1})}_{i^*,0}$.
	Then by the above construction for good sequences, one can construct a data structure for $(k_1,\ldots,k_{i^*-1})$ using $(i^*-1)w+1$ bits, and a data structure for $(k_{i^*+1},\ldots,k_B)$ using $(B-i^*)w+1$ bits.
	Hence, the total space is at most 
	\[
		\lceil\log B\rceil+\lceil\log 2\sigma\rceil+(i^*-1)w+1+(B-i^*)w+1\leq (B-1)w+\lfloor\log B\sigma\rfloor+5
	\] bits.
	By converting the extra (at most) $\lfloor\log B\sigma\rfloor+5$ bits to the spillover, one obtains a data structure using space $[K_{\textrm{bad}}]\times \{0,1\}^{(B-1)w}$ for 
	\begin{equation}\label{eqn_k_bad}
		K_{\textrm{bad}}=32B\sigma.
	\end{equation}

	To decode $T_i$, one first retrieves $i^*$.
	If $i<i^*$, $T_i$ can be decoded from the data structure for $(k_1,\ldots,k_{i^*-1})$ in constant time.
	If $i=i^*$, we have $T_{i^*}=T_{i^*-1}+\SUM(k_{i^*})$.
	The former term can be decoded in constant time from the data structure for the prefix, the latter term can be computed from $k_{i^*}$, which is explicitly stored in memory once $T_{i^*-1}$ is computed.
	If $i>i^*$, one can first compute $T_{i^*}$, then compute $T_{i}-T_{i^*}$ by querying the data structure for the suffix. 

	To decode $k_i$, if $i<i^*$ (or $i>i^*$), $k_i$ can be decoded from the data structure for the prefix (or the data structure for the suffix).
	If $i=i^*$, one computes $T_{i^*-1}$ using the above algorithm, and decodes the index of $k_{i^*}$ within $\cK^{(T_{i^*-1})}_{i^*,0}$, which determines $k_{i^*}$.
	This completes the construction when $i^*$ breaks the sequence into two good subsequences.

	\bigskip

	If the suffix $(k_{i^*+1},\ldots,k_B)$ is not good either, then there exists another index $i^*_2>i^*$ such that $k_{i^*_2}\in \cK^{(T_{i^*_2-1}-T_{i^*})}_{i^*_2-i^*,0}$, which also has size at most $2\sigma$.
	For a random input sequence, such case happens sufficiently rarely, so that the ``not-so-efficient'' solution of Lemma~\ref{lem_induct_lowprob} becomes affordable.
	We present the details below.

	\paragraph{There exist $i_1^*<i_2^*$ which are both in sets of size at most $2\sigma$.} 
	Suppose there exist $i_1^*$ and $i_2^*$ such that
	\begin{itemize}
		\item $k_{i_1^*}\in \cK^{(T_{i_1^*-1})}_{i_1^*,0}$;
		\item $k_{i_2^*}\in \cK^{(T_{i_2^*-1}-T_{i_1^*})}_{i_2^*-i_1^*,0}$.
	\end{itemize}
	Note that both $\cK^{(T_{i_1^*-1})}_{i_1^*,0}$ and $\cK^{(T_{i_2^*-1}-T_{i_1^*})}_{i_2^*-i_1^*,0}$ have size at most $2\sigma$.

	We first spend $2\lceil \log B\rceil$ bits to encode $i_1^*$ and $i_2^*$,  
	and another $2\lceil\log 2\sigma\rceil$ bits to encode $k_{i_1^*}$ within $\cK^{(T_{i_1^*-1})}_{i_1^*,0}$ and $k_{i_2*}$ within $\cK^{(T_{i_2^*-1}-T_{i_1^*})}_{i_2^*-i_1^*,0}$.
	$i_1^*$ and $i_2^*$ break the input sequence into three consecutive subsequences. 
	Next we apply Lemma~\ref{lem_induct_lowprob} to separately encode each subsequence: $(k_1,\ldots,k_{i^*_1-1})$, $(k_{i^*_1+1},\ldots,k_{i^*_2-1})$ and $(k_{i^*_2+1},\ldots,k_B)$.
	Hence, the total space usage is at most
	\begin{align*}
		&\ 2\log B+2\log 2\sigma+6+(B-2)(w+\log w) \\
		\leq&\ (B-1)w+2\log B\sigma +8+B\log w-w
	\end{align*}
	bits.
	By converting the extra (at most) $2\log B\sigma +8+B\log w-w$ bits to the spillover, one obtains a data structure using space $[K_{\textrm{lowprob}}]\times \{0,1\}^{(B-1)w}$ for 
	\begin{equation}\label{eqn_k_lowprob}
		K_{\textrm{lowprob}}= (B\sigma)^2 w^{B} 2^{-w+8} .
	\end{equation}
	To decode a $T_i$, one first retrieves $i^*_1$ and $i^*_2$.
	If $i<i^*_1$, $T_i$ can be decoded from the data structure for the first subsequence $(k_1,\ldots,k_{i^*_1-1})$.
	If $i=i^*_1$, one first decodes $T_{i^*_1-1}$, which determines the set $\cK^{(T_{i_1^*-1})}_{i_1^*,0}$.
	$k_{i^*_1}$ becomes retrievable in constant time, and $T_{i^*_1}$ can be computed as $T_{i^*_1-1}+\SUM(k_{i^*_1})$.
	When $i^*_1<i<i^*_2$, $T_i$ can be computed from $T_{i^*_1}$ and the data structure for $(k_{i^*_1+1},\ldots,k_{i^*_2-1})$.
	Decoding $T_i$ in the cases when $i=i^*_2$ or $i>i^*_2$ is similar.
	Likewise, one can also decode each $k_i$ in constant time.

	\bigskip
	\paragraph{The final data structure.} 
	To handle a general input sequence $(k_1,\ldots,k_B)$, we again glue the above three data structures together.
	By setting $K=K_{\textrm{good}}+K_{\textrm{bad}}+K_{\textrm{lowprob}}$, we obtain a data structure using space $[K]\times \{0,1\}^{(B-1)w}$ which allows one to decode each $k_i$ and each $T_i$ in constant time.
	Here, by Equation~\eqref{eqn_k_good}, \eqref{eqn_k_bad} and \eqref{eqn_k_lowprob}, we have
	\begin{align*}
		K&\leq 2^w+w^{O(B)}+32B\sigma+(B\sigma)^2 w^{B}2^{-w+8} \\
		&\leq 2^w+32B\sigma+(B\sigma)^2 w^{O(B)}2^{-w}, \\
		\intertext{which by the fact that $\sigma<2^w/n$, is at most}
		&\leq 2^w+B\sigma(32+Bw^{O(B)}/n) \\
		&\leq 2^w+33B\sigma,
	\end{align*}
	since $n\geq l\geq 2^{w/2B}-1$ and $B=w^{1/3}$.
	This proves the lemma.
\end{proof}

\begin{proof}[Proof of Claim~\ref{cl_part}]
By Fact~\ref{fact_size_bound} and Equation~\eqref{eqn_part}, we have
\[
	|\SUM^{-1}(T_i-T_{i-1})|\geq \binom{l}{T_i-T_{i-1}}\cdot 2^{-l+w}\geq \sum_{j=1}^{r_i}2^{e_{i,j}}\bOne_{X_{i,j}}(T_{i-1})\bOne_{Y_{i,j}}(T_{i}).
\]
Hence Item~(\ref{enum_part_b}) can be satisfied, e.g., by setting
\begin{align*}
	\cK^{(T_{i-1})}_{i,j}&:= \left\{t\textrm{-th element in }\SUM^{-1}(T_i-T_{i-1}): T_i\in T_{i-1}+l/2\pm \sqrt{l\log 1/\epsilon},\right. \\
	&\left. t\in \left[\sum_{a=1}^{j-1}2^{e_{i,a}}\bOne_{X_{i,a}}(T_{i-1})\bOne_{Y_{i,a}}(T_{i})+1,\sum_{a=1}^{j}2^{e_{i,a}}\bOne_{X_{i,a}}(T_{i-1})\bOne_{Y_{i,a}}(T_{i})\right]\right\}
\end{align*}
for $1\leq j\leq r_i$.

Now let $\cK^{(T_{i-1})}_{i,0}=[2^w+\sigma]\setminus \bigcup_{j=1}^{r_i} \cK^{(T_{i-1})}_{i,j}$.
By definition, for all $k_i$ such that $\SUM(k_i)\notin l/2\pm \sqrt{l\log 1/\epsilon}$, we have $k_i\in \cK^{(T_{i-1})}_{i,0}$ (since $\SUM(k_i)=T_i-T_{i-1}$).

To bound its size, we first observe that since $\sigma>n2^{w/2}$ and $n>l\geq 2^{w^{2/3}/2}-1$, by Lemma~\ref{lem_part}, for $T_{i-1}\in (i-1)(l/2\pm \sqrt{l\log 1/\epsilon})$, we have
\[
	\sum_{T_i} E(T_{i-1},T_i)\leq \sigma/4+w^{O(1)}2^{w/2}\leq \sigma/2.
\]
Hence,
\begin{align*}
	|\cK^{(T_{i-1})}_{i,0}|&= 2^w+\sigma-\sum_{j=1}^{r_i} |\cK^{(T_{i-1})}_{i,j}| \\
	&= 2^w+\sigma-\sum_{j=1}^{r_i}\sum_{T_i\in T_{i-1}+l/2\pm \sqrt{l\log 1/\epsilon}} 2^{e_{i,j}}\bOne_{X_{i,j}}(T_{i-1})\bOne_{Y_{i,j}}(T_{i}),\\
	\intertext{which by Equation~\eqref{eqn_part}, is equal to}
	&= 2^w+\sigma-\sum_{T_i-T_{i-1}\in l/2\pm \sqrt{l\log 1/\epsilon}}\left(\binom{l}{T_i-T_{i-1}}\cdot 2^{-l+w}-E_i(T_{i-1},T_i)\right) \\
	&\leq 2^w+\sigma-\sum_{s_i\in l/2\pm \sqrt{l\log 1/\epsilon}} \binom{l}{s_i}\cdot 2^{-l+w}+\sigma/2 \\
	&= 3\sigma/2+\sum_{s_i\notin l/2\pm \sqrt{l\log 1/\epsilon}} \binom{l}{s_i}\cdot 2^{-l+w} \\
	&\leq 3\sigma/2+2^{w+1}\epsilon \\
	&=2\sigma.
\end{align*}
\end{proof}

\subsection{Approximating binomial coefficients}\label{sec_binom_approx}
In this subsection, we prove Lemma~\ref{lem_part}.
We begin by approximating binomial coefficients in a small range.

\begin{lemma}\label{lem_approx}
	For any large integers $l$ and $d$, $0<\alpha\leq \frac{1}{2}$, such that $d\leq C\sqrt{\alpha l}$, there is a polynomial $P$ of degree $d^2$, such that
	\[
		\binom{l}{\alpha l+x+y}\cdot (1-\epsilon)\leq \binom{l}{\alpha l}\cdot \left(\frac{1-\alpha}{\alpha}\right)^{x+y} \cdot P(x+y)\leq \binom{l}{\alpha l+x+y}
	\]
	for all integers $x,y\in [0, C\cdot\sqrt{\alpha l}]$, a (small) universal constant $C>0$, and $\epsilon=2^{-d+8}$.

	Moreover, $P(x+y)$ can be written as a sum of $d^4+d^2+1$ nonnegative products:
	\[
		P(x+y)=\sum_{i=1}^{d^4+d^2+1}Q_i(x)R_i(y),
	\]
	where $Q_i(x),R_i(y)\geq 0$ for $x,y\in [0, C\cdot\sqrt{\alpha l}]$.
\end{lemma}

\begin{proof}
	First observe that for $t\in [0, 2C\cdot \sqrt{\alpha l}]$, we have
	\begin{align}
		\binom{l}{\alpha l+t}&=\binom{l}{\alpha l}\cdot\prod_{i=1}^{t}\frac{(1-\alpha)l-(i-1)}{\alpha l+i} \nonumber\\
		&=\binom{l}{\alpha l}\cdot\left(\frac{1-\alpha}{\alpha}\right)^t\cdot\prod_{i=1}^{t}\frac{1-\frac{i-1}{(1-\alpha)l}}{1+\frac{i}{\alpha l}}.\label{eqn_1}
	\end{align}
	It suffices to approximate the last factor $f(t):=\prod_{i=1}^{t}\frac{1-\frac{i-1}{(1-\alpha)l}}{1+\frac{i}{\alpha l}}$.

	For $|z|<1$, $1-z=e^{-\sum_{j\geq 1}j^{-1}\cdot z^j}$.
	Taking only the first $d-1$ terms, $e^{-\sum_{j=1}^{d-1} j^{-1}\cdot z^j}$, introduces a multiplicative error of
	\[
		e^{|\sum_{j\geq d}j^{-1} z^j|}\leq e^{\frac{1}{d(1-|z|)}|z|^{d}}.
	\]
	Applying the above approximation to $1-\frac{i-1}{(1-\alpha) l}$ and $1+\frac{i}{\alpha l}$ in \eqref{eqn_1}, we have
	\begin{align}
		f(t)&\approx \exp\left(-\sum_{i=1}^t\sum_{j=1}^{d-1} \frac{1}{j}\cdot \left(\frac{i-1}{(1-\alpha) l}\right)^j+\sum_{i=1}^t\sum_{j=1}^{d-1}\frac{1}{j}\cdot\left(-\frac{i}{\alpha l}\right)^j \right) \nonumber\\
		&=\exp\left(\sum_{j=1}^{d-1}\frac{1}{j}\cdot\sum_{i=1}^t\left(-(i-1)^j\cdot ((1-\alpha) l)^{-j}+i^j\cdot (-\alpha l)^{-j}\right)\right) \nonumber\\
		&=\exp\left(\sum_{j=1}^{d-1}\frac{1}{j}\cdot\left(-S_j(t-1)\cdot ((1-\alpha) l)^{-j}+S_j(t)\cdot (-\alpha l)^{-j}\right)\right) \label{eqn_2} \\
		&=: \exp\left(P_1(t)\right), \nonumber
	\end{align}
	where $S_j(t)=1^j+2^j+\cdots+t^j$ is a degree-$(j+1)$ polynomial of $t$.
	The approximation above has a total multiplicative error of at most
	\[
		\exp\left(\frac{2}{d}\sum_{i=1}^t\left(\left(\frac{i-1}{(1-\alpha) l}\right)^{d}+\left(\frac{i}{\alpha l}\right)^{d}\right)\right)\leq \exp\left(\frac{4t}{d}\cdot \left(\frac{t}{\alpha l}\right)^{d}\right),
	\]
	since $\alpha\leq 1/2$. When $t\leq 2C\cdot \sqrt{\alpha l}\leq \sqrt{\alpha l}$, this is at most
	\[
		\exp\left(\frac{4}{d(\alpha l)^{(d-1)/2}}\right)\leq 1+(\alpha l)^{-(d-1)/2}
	\]
	for large $d$.

	The exponent in \eqref{eqn_2} is a degree-$d$ polynomial of $t$, which we denote by $P_1(t)$.
	We then apply $e^z=1+\sum_{i\geq 1}z^i/i!$ to \eqref{eqn_2}. 
	Note that taking only the first $d$ terms in the sum introduces an \emph{additive} error of 
	\[
		\left|\sum_{i\geq d+1}z^i/i!\right|\leq \frac{2|z|^{d+1}}{(d+1)!}\leq 2^{-d}
	\]
	as long as $|z|\leq d/6$.
	The exponent in \eqref{eqn_2}, $P_1(t)$, is bounded by constants: Since $S_j(t)\leq t^{j+1}$, and
	\begin{align}
		\left|P_1(t)\right|&=\left|\sum_{j=1}^d\frac{1}{j}\cdot\left(-S_j(t-1)\cdot ((1-\alpha) l)^{-j}+S_j(t)\cdot (-\alpha l)^{-j}\right)\right| \nonumber\\
		&\leq \sum_{j=1}^d t^{j+1}\cdot ((1-\alpha) l)^{-j}+(\alpha l)^{-j}) \nonumber\\
		&\leq 2t\cdot \sum_{j\geq 1}(t/\alpha l)^j \nonumber\\
		&\leq 4t^2/\alpha l \nonumber\\
		&\leq 4.\label{eqn_bound_exp}
	\end{align}

	Hence for large $d$, this approximation produces a degree-$d^2$ polynomial $P_2(t)$, such that 
	\[
		f(t)\cdot (1-(\alpha l)^{-(d-1)/2})-2^{-d}\leq P_2(t)\leq f(t)\cdot (1+(\alpha l)^{-(d-1)/2})+2^{-d}.
	\]
	By a similar argument to \eqref{eqn_bound_exp}, we have $e^{-4}\leq f(t)\leq 1$.
	The additive errors can be translated into multiplicative errors,
	\[
		f(t)\cdot (1-(\alpha l)^{-(d-1)/2}-e^4\cdot 2^{-d})\leq P_2(t)\leq f(t)\cdot (1+(\alpha l)^{-(d-1)/2}+e^4\cdot 2^{-d}).
	\]
	By setting $P(t):=(1-\epsilon/2)P_2(t)$, we prove the first half of the lemma.

	\bigskip

	To prove the second half, the following equation is applied to transform a negative monomial $-x^ay^b$ (plus a large positive constant) to a sum of two nonnegative terms,
	\begin{equation}\label{eqn_trans}
		M^{a+b}-x^ay^b\equiv\frac{1}{2}(M^a-x^a)(M^b+y^b)+\frac{1}{2}(M^a+x^a)(M^b-y^b).
	\end{equation}
	Both terms are nonnegative when $x, y\in[0, M]$.
	It suffices to prove that the constant term of $P$ (or equivalently $P_2$) is large enough to accomplish this transformation simultaneously for all negative monomials.
	In particular, by expanding $P_2(x+y)=\sum_{a,b\geq 0} \beta_{a, b} x^a y^b$, it suffices to show
	\begin{equation}\label{eqn_goal}
		\beta_{0,0}\geq \sum_{a,b\geq 0,a+b\neq 0} |\beta_{a,b}| (C\sqrt{\alpha l})^{a+b}.
	\end{equation}
	Since if \eqref{eqn_goal} does hold, we will be able to rewrite
	\[
		P_2(x+y)=\left(\beta_{0,0}-\sum_{a,b\geq 0,a+b\neq 0} |\beta_{a,b}| (C\sqrt{\alpha l})^{a+b}\right)+\sum_{a,b\geq 0,a+b\neq 0}|\beta_{a,b}|\cdot\left((C\sqrt{\alpha l})^{a+b}+\sgn(\beta_{a,b}) x^ay^b\right)
	\]
	and apply \eqref{eqn_trans} to the later terms with $\beta_{a,b}<0$.
	The number of terms will be at most $d^4+d^2+1$, which will prove the second half of the lemma.

	In order to prove \eqref{eqn_goal}, let us first bound the coefficients of $S_j(t)$.
	By Faulhaber's formula (or Bernoulli's formula), 
	\[
		1^j+2^j+\cdots+t^j=S_j(t)=\frac{1}{j+1}\sum_{k=0}^{j}\binom{j+1}{k}\cdot B_k\cdot t^{j-k+1},
	\]
	where $B_k$ are the Bernoulli numbers:
	$B_0=1$; $B_1=1/2$; $B_k=0$ for all odd integers $k>1$;
	and for all even integers $k\geq 2$,
	\[
		B_k=(-1)^{k/2+1}\frac{2\cdot k!}{(2\pi)^k}\zeta(k),
	\]
	where $\zeta(z)=\sum_{m=1}^{\infty}m^{-z}$ is the Riemann zeta function~\cite{AWH05book}.

	The constant term in $P_1$ is $0$, since both $S_j(0)=0$ and $S_j(-1)=(-1)^{j+1}(S_j(1)-1)=0$.
	Hence, by definition, 
	\begin{equation}\label{eqn_lhs}
	\beta_{0,0}=P_2(0)=1.
	\end{equation}

	To upper bound the RHS of \eqref{eqn_goal}, consider the following operator:
	For polynomial $Q(x)=\sum_{i}\beta_ix^i$,
	\[
		F(Q, z):= \sum_{i}|\beta_i| z^i.
	\]
	Then for $z>0$,
	\begin{align*}
		F(S_j, z)&= \frac{1}{j+1}\sum_{k=0}^{j}\binom{j+1}{k}\cdot |B_k|\cdot z^{j-k+1} \\
		&\leq \sum_{k=0}^j \frac{2(j+1)^{k-1}}{(2\pi)^k}\cdot |\zeta(k)|\cdot z^{j-k+1}.
	\end{align*}
	Since $\lim_{k\rightarrow+\infty}\zeta(k)=1$, it is at most
	\[
		F(S_j, z)\leq \frac{z^{j+1}}{64C}\sum_{k=0}^j \frac{(j+1)^{k-1}}{(2\pi z)^k}.
	\]
	for sufficiently small constant $C>0$.
	When $j+1\leq 2\pi z$, it is at most $z^{j+1}/64C$.
	Therefore, by the fact that $j+1\leq d\leq C\sqrt{\alpha l}$, 
	\begin{align*}
		F(P_1, 2C\sqrt{\alpha l})&\leq \sum_{j=1}^{d-1}\frac{1}{j}\cdot \left(F(S_j, 2C\sqrt{\alpha l}+1)\cdot ((1-\alpha) l)^{-j}+F(S_j,2C\sqrt{\alpha l})\cdot (\alpha l)^{-j}\right) \\
		&\leq \sum_{j=1}^{d-1}\frac{1}{64Cj} \left((2C\sqrt{\alpha l}+1)^{j+1}\cdot ((1-\alpha) l)^{-j}+(\alpha l)^{-j}\right)\\
		&\leq \sum_{j=1}^{d-1}\frac{1}{32C} \left((4C\sqrt{\alpha l})^{j+1}\cdot (\alpha l)^{-j}\right)\\
		&\leq \frac{\sqrt{\alpha l}}{8}\sum_{j\geq 1}(4C/\sqrt{\alpha l})^{j} \\
		&\leq C.
	\end{align*}
	Thus, the RHS of \eqref{eqn_goal} is at most 
	\[
		\sum_{i=1}^d F(P_1, 2C\sqrt{\alpha l})^i/i!\leq e^{F(P_1, 2C\sqrt{\alpha l})}-1\leq e^C-1.
	\]
	Together with \eqref{eqn_lhs}, Equation~\eqref{eqn_goal} must hold for sufficiently small $C>0$.
	It completes the proof of the lemma.
\end{proof}

The above approximation can be extended to much larger ranges.
Let $C$ be the constant in Lemma~\ref{lem_approx}, we have the following.
\begin{lemma}\label{lem_binom}
	For any large even integer $l$, positive numbers $M_x$, $M_y$ and $\epsilon$, such that $l>8M_x, l>8M_y$ and $\epsilon>2^{-C\sqrt{l}/2+8}$, we can write $\binom{l}{l/2+x+y}\cdot 2^{-l}$ as
	\[
		\binom{l}{l/2+x+y}\cdot 2^{-l}=E(x, y)+\sum_{i=1}^r \tilde{Q}_i(x)\tilde{R}_i(y)
	\]
	for all integers $x\in[-M_x,M_x]$ and $y\in [-M_y, M_y]$, such that 
	\begin{enumerate}[a)]
		\item $E(x, y)\geq 0$ and for every integer $x\in[-M_x, M_x]$, $\sum_{y=-M_y}^{M_y} E(x, y)\leq \epsilon$;
		\item for all $i\in [r]$ and $x\in[-M_x,M_x], y\in[-M_y, M_y]$, $0\leq \tilde{Q}_i(x),\tilde{R}_i(y)\leq 1$;
		\item $r\leq O((M_xM_y/l)\log^4 1/\epsilon)$.
	\end{enumerate}
\end{lemma}

\begin{proof}
	The idea is to first partition the domain $[-M_x, M_x]\times [-M_y, M_y]$ into rectangles, so that in each rectangle $[a_x, b_x]\times [a_y, b_y]$, we have
	\[
		\binom{l}{l/2+x+y}=\binom{l}{(l/2+a_x+a_y)+(x-a_x)+(y-a_y)}
	\]
	and
	\[
		\binom{l}{l/2+x+y}=\binom{l}{(l/2-b_x-b_y)+(b_x-x)+(b_y-y)},
	\]
	which allows us to apply Lemma~\ref{lem_approx} as long as both sides of the rectangle are not too large.
	Finally, we put all error terms from these rectangles into $E(x, y)$.
	In the following, we present this construction in detail.

	\paragraph{Approximation in subrectangles.}
	We partition $[-M_x, M_x]\times [-M_y, M_y]$ into $O(M_xM_y/l)$ rectangles of size at most $C\sqrt{l}/2\times C\sqrt{l}/2$.
	For each rectangle $\cR_k=[a_x^{(k)},b_x^{(k)}]\times[a_y^{(k)},b_y^{(k)}]$, if $a_x^{(k)}+a_y^{(k)}\leq 0$, we apply Lemma~\ref{lem_approx} to approximate the binomial coefficient
	\[
		\binom{l}{(l/2+a_x^{(k)}+a_y^{(k)})+(x-a_x^{(k)})+(y-a_y^{(k)})}
	\]
	with $d=\log (2^8/\epsilon)$.
	The premises of the lemma are satisfied:
	\begin{itemize}
		\item $\alpha=1/2+a_x^{(k)}/l+a_y^{(k)}/l\leq 1/2$;
		\item since $l/2+a_x+a_y\geq l/2-M_x-M_y\geq l/4$, we have $\alpha\geq 1/4$ and $C\sqrt{l}/2\leq C\sqrt{\alpha l}$, i.e., $(x-a_x^{(k)}),(y-a_y^{(k)})\in[0, C\sqrt{\alpha l}]$;
		\item $d=\log (2^8/\epsilon)\leq C\sqrt{l}/2\leq C\sqrt{\alpha l}$.
	\end{itemize}
	Otherwise, $b_x^{(k)}+b_y^{(k)}$ must be greater than zero, and we apply Lemma~\ref{lem_approx} to approximate
	\[
		\binom{l}{(l/2-b_x^{(k)}-b_y^{(k)})+(b_x^{(k)}-x)+(b_y^{(k)}-y)},
	\]
	which is equal to $\binom{l}{l/2-x-y}=\binom{l}{l/2+x+y}$.
	Similarly, one can verify in this case, all premises of the lemma are also satisfied.

	Therefore, we obtain an approximation of $\binom{l}{l/2+x+y}$ in rectangle $\cR_k$,
	\[
		\sum_{i=1}^{r_k}\left(\binom{l}{\alpha l}\cdot\left(\frac{1-\alpha}{\alpha}\right)^x Q_i^{(k)}(x)\right)\cdot \left(\left(\frac{1-\alpha}{\alpha}\right)^y R_i^{(k)}(y)\right),
	\]
	where $\alpha l=l/2+a_x^{(k)}+a_y^{(k)}$ (or $l/2-b_x^{(k)}-b_y^{(k)}$, depending on their values), and $r_k=O(d^4)=O(\log^4 1/\epsilon)$.
	Next we make both $Q$ and $R$ bounded by $1$ and vanish outside $\cR_k$, by scaling them by the following factor and multiplying by the indicator functions of the two sides correspondingly.
	Let 
	\[
		C_i^{(k)}=\max_{y\in [a_y^{(k)},b_y^{(k)}]}\left(\frac{1-\alpha}{\alpha}\right)^y R_i^{(k)}(y),
	\]
	and let
	\[
		\tilde{Q}_i^{(k)}(x):=C_i^{(k)}\cdot 2^{-l}\binom{l}{\alpha l}\left(\frac{1-\alpha}{\alpha}\right)^x Q_i^{(k)}(x)\cdot \mathbf{1}_{[a_x^{(k)}, b_x^{(k)}]}(x)
	\]
	and
	\[
		\tilde{R}_i^{(k)}(y):=\frac{1}{C_i^{(k)}}\left(\frac{1-\alpha}{\alpha}\right)^y R_i^{(k)}(y)\cdot \mathbf{1}_{[a_y^{(k)},b_y^{(k)}]}(y),
	\]
	where $\mathbf{1}_{S}(z)$ is equal to $1$ if $z\in S$, and $0$ if $z\notin S$.

	Now consider $W^{(k)}(x, y):=\sum_{i=1}^{r_k}\tilde{Q}_i^{(k)}(x)\cdot \tilde{R}_i^{(k)}(y)$, we have
	\begin{enumerate}[i)]
		\item for $(x, y)\in \cR_k$, $2^{-l}\binom{l}{l/2+x+y}(1-\epsilon)\leq W^{(k)}(x, y)\leq 2^{-l}\binom{l}{l/2+x+y}$;
		\item for $(x, y)\notin \cR_k$, $W^{(k)}(x, y)=0$;
		\item $\tilde{Q}_i^{(k)}(x), \tilde{R}_i^{(k)}(y)\geq 0$;
		\item by definition, $\max_{y\in [a_y^{(k)},b_y^{(k)}]} \tilde{R}_i^{(k)}(y)=1$, and 
		\begin{align*}
			\tilde{Q}_i^{(k)}(x)&\leq \min_{y\in [a_y^{(k)},b_y^{(k)}]}\left\{W^{(k)}(x, y)/\tilde{R}_i^{(k)}(y)\right\} \\
			&\leq \left(\max_{y\in [a_y^{(k)},b_y^{(k)}]}\tilde{R}_i^{(k)}(y)\right)^{-1} \\
			&=1.
		\end{align*}
	\end{enumerate}

	\paragraph{Merging the subrectangles.}
	Finally, let 
	\[
		W(x, y):=\sum_{k\leq O(M_xM_y/l)} W^{(k)}(x, y)=\sum_{k\leq O(M_xM_y/l)}\sum_{i=1}^{r_k}\tilde{Q}_i^{(k)}(x)\tilde{R}_i^{(k)}(y)
	\]
	and
	\[
		E(x, y):=\binom{l}{l/2+x+y}\cdot 2^{-l}-W(x, y).
	\]

	In the following, we show this construction indeed has the claimed properties:
	\begin{enumerate}[a)]
		\item by Item i) and ii) above and the fact that $\{\cR_k\}$ is a partitioning of $[-M_x, M_x]\times [-M_y, M_y]$, we have $E(x, y)\geq 0$; for every integer $x\in [-M_x, M_x]$, we have
		\[
			\sum_{y=-M_y}^{M_y} E(x, y)\leq \sum_{y=-M_y}^{M_y}\epsilon 2^{-l}\binom{l}{l/2+x+y}\leq \epsilon;
		\]
		\item by Item iii) and iv) above, we have $0\leq \tilde{Q}^{(k)}_i(x),\tilde{R}^{(k)}_i(x)\leq 1$;
		\item $r=\sum_{k\leq O(M_xM_y/l)}r_k=O((M_xM_y/l)\log^4 1/\epsilon)$.
	\end{enumerate}
	This proves the lemma.

\end{proof}

Finally, we are ready to prove Lemma~\ref{lem_part}.
\begin{restate}[Lemma~\ref{lem_part}]
	For any large even integer $l$, positive numbers $M_x, M_y$ and $\epsilon$, such that $l>8M_x$, $l>8M_y$ and $\epsilon>2^{-C\sqrt{l}/2+8}$, we have
	\[
		\binom{l}{l/2 +x+y}\cdot 2^{-l+w}=E(x, y)+\sum_{i=1}^r 2^{e_i} \bOne_{X_i}(x)\bOne_{Y_i}(y)
	\]
	for all integers $x\in [-M_x, M_x]$ and $y\in [-M_y, M_y]$, such that 
	\begin{enumerate}[a)]
		\item $E(x, y)\geq 0$ and for every $x\in[-M_x, M_x]$, $\sum_{y=-M_y}^{M_y} E(x, y)\leq \epsilon 2^w+4rM_y2^{w/2}$;
		\item for every $i\in [r]$, $e_i\geq 0$ is an integer, $X_i\subseteq [-M_x, M_x]$, $Y_i\subseteq [-M_y, M_y]$ are sets of integers;
		\item $r\leq O((M_xM_y/l)w^2\log^4 (1/\epsilon))$.
	\end{enumerate}
\end{restate}

\begin{proof}
By Lemma~\ref{lem_binom} and multiplying both sides by $2^w$, we have
\begin{align}
	\binom{l}{l/2+x+y}\cdot 2^{-l+w}&=E_0(x, y)2^w+\sum_{i=1}^{r_0} (2^{w/2}\tilde{Q}_i(x))\cdot (2^{w/2}\tilde{R}_i(y)) \nonumber\\
	&=E(x, y)+\sum_{i=1}^{r_0} \lfloor2^{w/2}\tilde{Q}_i(x)\rfloor\cdot \lfloor2^{w/2}\tilde{R}_i(y)\rfloor,\label{eqn_part_2}
\end{align}
for $E(x, y)\geq 0$ and
\[
	E(x, y)\leq E_0(x, y)2^w+\sum_{i=1}^{r_0}\left(2^{w/2}\tilde{Q}_i(x)+2^{w/2}\tilde{R}_i(y)\right) \leq E_0(x, y)2^w+2r_02^{w/2}.
\]
Therefore, we have
$\sum_{y=-M_y}^{M_y} E(x, y)\leq \epsilon 2^{w}+4r_0M_y2^{w/2}$
for every $x\in[-M_x,M_x]$.
It proves Item a).

Next, we write each $\lfloor2^{w/2}\tilde{Q}_i(x)\rfloor$ and $\lfloor2^{w/2}\tilde{R}_i(y)\rfloor$ in binary representation.
Let
\[
	\lfloor2^{w/2}\tilde{Q}_i(x)\rfloor=\sum_{j=0}^{w/2-1} 2^j\cdot \bOne_{X_{i,j}}(x)
\]
and
\[
	\lfloor2^{w/2}\tilde{R}_i(y)\rfloor=\sum_{j=0}^{w/2-1} 2^j\cdot \bOne_{Y_{i,j}}(y),
\]
where $X_{i,j}\ni x$ (resp. $Y_{i,j}\ni y$) if and only $\lfloor2^{w/2}\tilde{Q}_i(x)\rfloor$ (resp. $\lfloor2^{w/2}\tilde{R}_i(y)\rfloor$) has a ``1'' in the $j$-th bit in its binary representation.

By expanding Equation~\eqref{eqn_part_2}, we have
\[
	\binom{l}{l/2+x+y}\cdot 2^{-l+w}=E(x, y)+\sum_{i=1}^{r_0}\sum_{j_1=0}^{w/2-1}\sum_{j_2=0}^{w/2-1}2^{j_1+j_2}\bOne_{X_{i,j_1}}(x)\bOne_{Y_{i,j_2}}(y).
\]

Hence, $r=O(r_0 w^2)\leq O((M_xM_y/l)w^2\log^4 1/\epsilon)$.

\end{proof}

\subsection{Using standard word operations}\label{sec_ram}

The above data structure assumes that the computational model allows one to compute arbitrary functions on $O(w)$-bit input (that are hard-wired in the solution) in constant time.
To only use standard word operations, one nature idea is to precompute all such functions needed at preprocessing time, and store a lookup table in memory.

By examining the data structure, one may verify that the only part that uses non-standard word operations is Lemma~\ref{lem_induct}, combining the $B$ spillovers.
This subroutine is applied at $t$ different levels, for $\sigma=(34B)^i\cdot n2^{w/2}$ ($i=0,\ldots,t-1$).
In each application, there are $2B$ different possible queries (decoding one of $k_1,\ldots,k_B$ or $T_1,\ldots,T_B$).
The query time is bounded by a universal constant $c_q$.
Hence, the whole query algorithm, which is a decision tree, can be encoded by a lookup table of size
\[
	t\cdot 2B\cdot \sum_{k=0}^{c_q} 2^{kw}\cdot O(w)=O(tB2^{c_qw}\cdot w)
\]
bits.

Since $w$ is required to be at least $7\log n$, storing the entire lookup table is unaffordable.
However, one may use a standard trick to decrease $w$ for self-reducible problems.
Given an input of $n$ bits, we evenly partition the input into blocks of size $n'=n^{1/8c_q}$.
For each block, we apply Theorem~\ref{thm_upper_cell} with $w'=\frac{7\log n}{8c_q}$, and construct a data structure using
\[
	n'+\lceil\frac{n'}{w'^{\Omega(t)}}\rceil\leq n'+\frac{n'}{(\log n)^{\Omega(t)}}+1
\]
bits.
In addition, we also store the lookup table for the whole query algorithm using
\[
	O(tB2^{c_qw'}\cdot w')=\tilde{O}(n^{7/8})
\]
bits.
Note that this lookup table is \emph{shared} among all blocks, and hence only one copy needs to be stored.
Finally, the total space usage is
\[
	n+\frac{n}{(\log n)^{\Omega(t)}}+O(n^{1-1/8c_q}).
\]
This proves Theorem~\ref{thm_upper_ram}.
\begin{restate}[Theorem~\ref{thm_upper_ram}]
	Given a 0-1 array of length $n$ for sufficiently large $n$, for any $t\geq 1$, one can construct a succinct data structure using
	\[
		n+ \frac{n}{(\log n)^{\Omega(t)}}+n^{1-c}
	\]
	bits of memory supporting rank queries in $O(t)$ time, in a word RAM with word-size $w=\Theta(\log n)$, for some universal constant $c>0$.
\end{restate}
\section{Discussion and Open Questions}\label{sec_concl}
The sibling of the rank query is \emph{select}, which asks ``where is the $k$-th one in the input array?''
Most previous succinct rank data structures support both rank and select at the same time, and the lower bounds~\cite{PV10} also apply to select queries.
It is a natural question to ask whether our data structure also generalizes to both rank and select.

Several previous solutions also apply to sparse inputs (e.g., \cite{Pagh01}, \cite{GRR08}, \cite{Pat08}, etc).
That is, we are given an array of $n$ bits with $m$ ones, and would like to design a data structure with space close to $\log \binom{n}{m}$.
However, generalizing our data structure to such instances seems more complicated, may require massive technical manipulation.

\bigskip

Most cell-probe lower bound proof techniques also apply to the \emph{information cell-probe model}, defined in Section~\ref{sec_overview}.
However, some data structure problems are trivial in the information cell-probe model, while their complexities in the cell-probe model are unclear (e.g., succinct dictionary).
If one believes such a problem does not admit trivial cell-probe data structures, then proving any lower bound requires distinguishing the two models, and likely needs a new technique.
One might hope to apply the breakthrough in information and communication~\cite{GKR14} to achieve such separation.
Despite the close connections between data structures and communication complexity, separating cell-probe and information cell-probe is not entirely equivalent to separating information and communication complexity.
One reason is that only the memory side changes the measure from worst-case communication to information, while the data structure side still sends a message of fixed length.
It is unclear to us whether similar separations can be established in communication complexity.

\paragraph{Acknowledgments} The author would like to thank Jiantao Jiao and Tengyu Ma for helpful discussions.

\bibliography{ref}
\bibliographystyle{alpha}

\appendix
\section{Proof of Lemma~\ref{lem_change_base}}\label{app_change_base}
Observe that given two elements $x_i\in [M_i]$ and $x_j\in [M_j]$, we may treat the pair as a number from $[M_i\cdot M_j]$ by mapping $(x_i, x_j)\mapsto x_i M_j+x_j$.
Since initially each $M_i\leq 2^w$, by repeatedly merging elements in this way, we may assume the universe size of each element is between $2^{3w}$ and $2^{6w}$, except for the last element, which may be from a set of size smaller than $2^{3w}$.

Now the task becomes to encode a sequence $(y_1,\ldots,y_{B'})\in[N_1]\times\cdots\times[N_{B'}]$ such that
\begin{itemize}
	\item for $i=1,\ldots,B'-1$, $2^{3w}<N_i\leq 2^{6w}$, and
	\item $N_{B'}\leq 2^{6w}$.
\end{itemize}
We first inductively calculate and hard-wire the following constants:
\begin{align}
	V_0&=1 \nonumber\\
	\intertext{and for $i\geq $1, }
	w_i&=\lfloor\log (V_{i-1})+1.5w\rfloor\label{eqn_wi}\\
	U_i&=\lfloor\frac{2^{w_i}}{V_{i-1}}\rfloor\label{eqn_ui}\\
	V_i&=\lceil\frac{N_i}{U_i}\rceil\label{eqn_vi}
\end{align}
Then we break each element $y_i\in[N_i]$ into a pair $(u_i, v_i)$ such that $u_i\in [U_i]$ and $v_i\in [V_i]$ due to \eqref{eqn_vi}, e.g., $u_i:=\lfloor y_i/V_i \rfloor$ and $v_i=y_i\!\!\mod V_i$.
Then for $1\leq i\leq B'-1$, we combine $(v_{i-1}, u_i)$ into an element $z_i\in [2^{w_i}]$ due to \eqref{eqn_ui}, e.g., $z_i=v_{i-1}\cdot U_i+u_i$.

The data structure will explicitly store $z_1,\ldots,z_{B'-1}$ using $w_1+\cdots+w_{B'-1}$ bits.
For the last two elements $(v_{B'-1}, y_{B'})$, we combine them into an element $z_{B'}\in [V_{B'-1}\cdot N_{B'}]$, and store $z_{B'}$ using $m-(w_1+\cdots+w_{B'-1})$ bits and a spillover of size 
\[
	K=\lceil V_{B'-1}\cdot N_{B'}\cdot 2^{(w_1+\cdots+w_{B'-1})-m}\rceil.
\]
To decode a $y_i$, it suffices to retrieve $z_{i}$ and $z_{i+1}$, which determines $u_i$ and $v_i$ respectively.
Since $w_i=O(w)$, it takes constant time.

The data structure uses $m$ bits of memory. 
It remains to show that $K$ is at most $\left\lceil 2^{-m}\cdot\prod_{i=1}^B M_i\right\rceil+1$.
By Equation~\eqref{eqn_wi}, \eqref{eqn_ui} and \eqref{eqn_vi}, we have
\[
	U_i=\Theta(2^{w_i}/V_i)= \Theta(2^{1.5w}),
\]
and
\[
	V_i=\Theta(N_i/U_i)\geq \Omega(2^{1.5w}).
\]
Thus, by \eqref{eqn_ui} again, we have
\[
	2^{w_i}\leq (U_i+1)V_{i-1}\leq U_iV_{i-1}\cdot (1+O(2^{-1.5w})),
\]
and by \eqref{eqn_vi}, we have
\[
	U_iV_i\leq N_i\cdot (1+O(2^{-1.5w})).
\]

Finally, by definition, we have
\begin{align*}
	K&=\lceil2^{-m}\cdot V_{B'-1}N_{B'}\prod_{i=1}^{B'-1}2^{w_i}\rceil \\
	&\leq \lceil2^{-m}\cdot V_{B'-1}N_{B'}\prod_{i=1}^{B'-1}U_iV_{i-1}\cdot (1+O(2^{-1.5w}))\rceil \\
	&\leq \lceil2^{-m}\cdot N_{B'}\cdot (1+O(B'\cdot2^{-1.5w}))\prod_{i=1}^{B'-1}U_iV_i\rceil \\
	&\leq \lceil2^{-m}\cdot (1+O(B'\cdot2^{-1.5w}))\prod_{i=1}^{B'}N_i\cdot (1+O(2^{-1.5w}))\rceil \\
	&\leq \lceil2^{-m}\cdot (1+O(B\cdot2^{-1.5w}))\prod_{i=1}^{B}M_i\rceil.
\end{align*}
The last inequality is because $\prod_{i=1}^{B} M_i=\prod_{i=1}^{B'}N_i$.
Since $m\geq \sum_{i=1}^{B}\log M_i-w$ and $B=o(2^{w/2})$, we have
\[
	K\leq \lceil2^{-m}\prod_{i=1}^{B}M_i+1\rceil= \lceil2^{-m}\prod_{i=1}^{B}M_i\rceil+1.
\]
This proves the lemma.

\end{document}